\def\year{2018}\relax
\documentclass[letterpaper]{article} %DO NOT CHANGE THIS
\usepackage{aaai18}  %Required
\usepackage{times}  %Required
\usepackage{helvet}  %Required
\usepackage{courier}  %Required
\usepackage{url}  %Required
\usepackage{graphicx}  %Required
\usepackage[boxed]{algorithm}
\usepackage[noend]{algorithmic}
\usepackage{xspace}
\usepackage{amsmath}
\usepackage{amssymb}
\usepackage{subfigure}
\usepackage{hyperref}

\title{Comparing Population Means under Local Differential Privacy:\\with Significance and Power}

\author{Bolin Ding,~ Harsha Nori,~ Paul Li,~ Joshua Allen\\
\{bolind, hanori, paul.li, joshuaa\}@microsoft.com\\
Microsoft, One Microsoft Way, Redmond, WA 98052}

\begin{document}

\newtheorem{theorem}{Theorem}
\newtheorem{definition}{Definition}
\newtheorem{lemma}{Lemma}
\newtheorem{corollary}{Corollary}
\newtheorem{proposition}{Proposition}

% You can add more of these if it is helpful.
\newcommand{\qed}{\hfill $\framebox(6,6){}$}

% Use the proof environment when the proof immediately follows the corresponding theorem or lemma.
\newenvironment{proof}{\par{\noindent \bf Proof:}}{\qed \par \smallskip}

% Use the proofof environment when the proof appears later.
\newenvironment{proofof}[1]{\par{\noindent \bf Proof of #1:}}{\qed \par}

% Use the proofsketch environment for less formal proof ideas.
\newenvironment{proofsketch}{\par{\noindent \bf Proof Sketch:}}{\qed \par}

\newcommand{\eop}{{\hspace*{\fill}$\Box$\par}}

\newcommand{\stitle}[1]{\noindent{\bf #1}}

\newcommand{\cequ}{Equation~}
\newcommand{\cequs}{Equations~}
\newcommand{\csec}{Section~}
\newcommand{\csecs}{Sections~}
\newcommand{\cexa}{Example~}
\newcommand{\cexas}{Examples~}
\newcommand{\cdef}{Definition~}
\newcommand{\cdefs}{Definitions~}
\newcommand{\cthm}{Theorem~}
\newcommand{\cthms}{Theorems~}
\newcommand{\clem}{Lemma~}
\newcommand{\clems}{Lemmas~}
\newcommand{\cprop}{Proposition~}
\newcommand{\cprops}{Propositions~}
\newcommand{\calg}{Algorithm~}
\newcommand{\calgs}{Algorithms~}
\newcommand{\cfig}{Figure~}
\newcommand{\cfigs}{Figures~}

\newcommand{\ie}{i.e.\xspace}
\newcommand{\eg}{e.g.\xspace}
\newcommand{\st}{s.t.\xspace}
\newcommand{\etc}{etc.\xspace}

\newcommand{\pr}[1]{{\bf Pr}\!\left[#1\right]\xspace}
\newcommand{\ep}[1]{{\bf E}\!\left[#1\right]\xspace}
\newcommand{\vr}[1]{{\bf Var}\!\left[#1\right]\xspace}
\newcommand{\bigoh}[1]{{\rm O}\!\left(#1\right)\xspace}
\newcommand{\bigtheta}[1]{{\rm \Theta}\!\left(#1\right)\xspace}
\newcommand{\bigomega}[1]{{\rm \Omega}\!\left(#1\right)\xspace}

\newcommand{\binomial}{{{\mathfrak B}\xspace}}
\newcommand{\normal}{{{\mathfrak N}\xspace}}
\newcommand{\algo}{{{\mathfrak A}\xspace}}
\newcommand{\algobit}[1]{{{\mathfrak M}_{#1}\xspace}}
\newcommand{\test}{{{\mathfrak T}\xspace}}
\newcommand{\testest}{\test^{\sf est}\xspace}
\newcommand{\testbin}{\test^{\sf bin}\xspace}
\newcommand{\testmix}{\test^{\sf mix}\xspace}
\newcommand{\eps}{{{\varepsilon}\xspace}}
\newcommand{\one}{{\bf 1}\xspace}
\newcommand{\pa}{{\bf A}\xspace}
\newcommand{\pb}{{\bf B}\xspace}
\newcommand{\px}{{\bf X}\xspace}
\newcommand{\py}{{\bf Y}\xspace}

\maketitle

\begin{abstract}
A statistical hypothesis test determines whether a hypothesis should be rejected based on samples from populations.
In particular, randomized controlled experiments (or A/B testing) that compare population means using, \eg, $t$-tests, have been widely deployed in technology companies to aid in making data-driven decisions.
Samples used in these tests are collected from users and may contain sensitive information.
Both the data collection and the testing process may compromise individuals' privacy.
In this paper, we study how to conduct hypothesis tests to compare population means while preserving privacy. We use the notation of {\em local differential privacy} (LDP), which has recently emerged as the main tool to ensure each individual's privacy without the need of a {\em trusted data collector}. We propose LDP tests that inject noise into every user's data in the samples before collecting them (so users do not need to trust the data collector), and draw conclusions with bounded type-I (significance level) and type-II errors ($1-$ power). Our approaches can be extended to the scenario where some users require LDP while some are willing to provide exact data. We report experimental results on real-world datasets to verify the effectiveness of our approaches.
\end{abstract}

\newcommand{\appendixtext}{{the appendix}\xspace}

\section{Introduction}
\label{sec:intro}

Randomized controlled experiments (or A/B testing) and hypothesis tests are used by many companies, \eg, Google, Facebook, Amazon, and Microsoft, to design and improve their products and services
\cite{kdd:TangAOM10,ics:Panger16,url:KohaviR04,kdd:KohaviDFLWX12}.
These statistical techniques base business decisions on samples of actual customer data collected during experiments to draw more informed conclusions. However, such data samples usually contain sensitive information, \eg, usage statistics of certain apps or services; in order to meet users' privacy expectations and tightening privacy regulations (\eg, European GDPR law), ensuring that these experiments and tests do not breach the privacy of individuals is an important problem.
%
% with privacy concerns mounting from users \cite{url:facebook} and regulators \cite{url:eucommission}, experimentation will need to be privacy compliant. 

Differential privacy (DP) \cite{TCC06} has emerged as a standard for the privacy guarantees, and been used by, \eg, Apple \cite{url:apple}, Google \cite{ccs:ErlingssonPK14}, and Uber \cite{url:uber}.
In a well-studied DP model used by, \eg, \cite{icml:RogersVLG16}, users trust and send exact data to a data collector, who then injects noise in the testing process to ensure DP. However, this model is not applicable in our setup, as users may not trust the data collector (\eg, a tech company) due to potential hacks and leaks \cite{url:hacks}, and prefer not to have unprivatized data leave their devices.
%
%, regardless of assurances from the data collector.
%
Therefore, we adopt the {\em local model of differential privacy} (LDP) \cite{focs:DuchiJW13}.
Under LDP, users do not need to trust the data collector. Before sent to the data collector, each user's data is privatized by a randomized algorithm with the property that the likelihood of any specific output of the algorithm varies little with the input, \ie, the exact data.
% , and thus stronger plausible deniability is provided to each user.

In this paper, we study how to conduct hypothesis tests to compare population means (\eg, in A/B testing), while ensuring LDP for each user. We focus on the class of $t$-tests when presenting our solutions -- they can be easily extended for $Z$-tests if populations follow Normal distributions.

% One likely approach to be privacy compliant, which also has implications for experimentation, is local differential privacy (LDP), where users' data are obfuscated on their devices prior to uploading. Local differential privacy is used by Apple in iOS 10.0 \cite{url:apple}, Google \cite{ccs:ErlingssonPK14}, Uber \cite{url:uber,corr:JohnsonNS17}, and is in the latest iteration of the Windows operating system. Experimentation with locally differential private data presents challenges because the data obfuscation affects both statistical power and statistical testing methods. 

An A/B test splits users randomly into two populations, to give them two different experiences, a {\em control} and a {\em treatment}, respectively, and then tests for differences between the two population means in a measure of interest (clicks, usage, monetization, \etc). A {\em null hypothesis} $H_0$ is that the two population means are equal or differ by a fixed constant. Statistical tests (\eg, $t$-tests) are used to determine whether the null hypothesis should be rejected based on random samples from the populations. To measure errors in the conclusions, {\em type-I error} is the probability of falsely rejecting $H_0$ when it is true, and {\em type-II error}, or complement of {\em statistical power}, is the probability of failing to reject $H_0$ when it is false. We want a test to have type-I error bounded by a pre-specified threshold, called {\em significance level}, and have high {power}.

A typical test has three common steps: 1) compute the observed value of a {\em test statistic} from samples; 2) calculate the {\em p-value}, \ie, the probability, under $H_0$, of a test statistic being more extreme than the observed one; 3) reject $H_0$ if and only if the p-value is less than the significance level.

% Through randomization, all the other factors that may lead to differences in the measure are balanced between the two groups (thus any difference can be attributed to the change between control and treatment). Subsequently, a straight-forward comparison of the measure (\eg its mean) in the two groups can be done using Statistical tests (\eg, $z$-test and $t$-test).

\stitle{Challenges and our contributions.} In $t$-tests (as well as $Z$-tests when population variances are unknown) that compare population means, sample means and sample variances are the essential terms in the test statistic to be computed in step 1). While there are several approaches, \eg, \cite{focs:DuchiJW13}, to estimate means, there is no known technique to estimate sample variances under LDP. In fact, we show that {\em any} estimator to variances based on a previous LDP mechanism \cite{nips:DingKY17} has a very large worst-case error (\cprop\ref{prop:varhardness}).

Our first approach, called {\em estimation-based LDP test}, is based on a seemingly direct idea.
We propose a new LDP mechanism to estimate sample variances, using which we obtain an estimation of the observed test statistic in step 1) to calculate the p-value and then draw the conclusion.

One of the most important goals of hypothesis testing is to control the probability of drawing a false conclusion in terms of type-I and type-II errors. 
While our new LDP variance estimator is of independent interest, the first approach is unsatisfying in achieving this goal, especially when the size of the data domain is large. As errors in both the estimator to sample means and the one to variances are proportional to the domain size, it is hard to bound the error in estimating the test statistic in step 1), so there is no theoretical guarantee on type-I and type-II errors in our first approach.

% (even when they are bounded as tightly as in \cite{focs:DuchiJW13} and \cprop\ref{prop:2bitvar} in this paper)

% Furthermore, prior to launch, experimenters typically pre-determine the number of users needed in the groups to detect differences of specific sizes with specific levels of statistical confidence in order to minimize exposure (\ie, possible detriment effects) and costs, based on some prior knowledge about the measure (\eg mean and standard deviation). Differential privacy introduces random noise into the data and transforms the value of the measure (typically from a real values to a bit), thus impacting the statistical test and the power calculation. 

% We note that {\em differential privacy} (DP) was first proposed by \cite{TCC06}, with the most well-studied DP model being the {\em global model}, where a trusted data collector first collects all users' data and then injects noise in the analysis step. However, in practice, users are unlikely to want unprivatized data to leave their devices, regardless of assurances from the collecting organization. Hacks \cite{url:hacks}, leaks \cite{url:leaks}, and governmental intervention \cite{url:government}.
%
% but the law is actually here \url{https://www.gpo.gov/fdsys/granule/CRI-2015/CRI-2015-FOREIGN-INTELLIGENCE-SURVEILLANC-C6A5FD/content-detail.html}) have all compromised users' data privacy in the past.
%
% Consequently, the {\em local model} (LDP) of differential privacy \cite{evfimievski2003limiting,nips:DuchiWJ13,stoc:BassilyS15}, which obfuscates user data prior to sending to the data collector is preferable.

The second approach we propose, called {\em transformation-based LDP test}, aims to provide {\em an upper bound of type-II error} at a pre-specified significance level, \ie, {\em a hard constraint of type-I error}. The main idea is to look into the relationship between the original distribution of a population and the distribution on the outputs of the LDP data-collection algorithm on users' data, called {\em transformed distribution}. Instead of estimating sample means and sample variances under LDP, we directly conduct tests based on LDP samples from the transformed distributions -- the conclusion can then be translated into a conclusion of the test on the original population (\ie, rejecting or accepting $H_0$).
The upper bound on type-II error during A/B testing is critical in estimating the number of users needed in the samples to detect significant differences between the control and the treatment populations.
We derive such an estimation of sample sizes needed to reduce type-II error below a threshold at the specified significance level.
This approach can be extended to a hybrid-privacy scenario where some users require LDP while some are willing to provide exact data. 

Experiments are conducted on real datasets to verify our theoretical results and the effectiveness of our approaches.

\stitle{Related work.} There are a long line of works on hypothesis testing under the DP model {\em with} a {\em trusted} data collector, with genome-wide association studies as a primary application. In this setup, the data collector receives exact samples from users.
The first type of approaches inject noise into aggregates (or marginal tables) of data to ensure DP, and compute or estimate the test statistic in step 1) from these noisy aggregates \cite{psd:FienbergRY10,psd:KarwaS12,kdd:JohnsonS13,annstat:KarwaS2016}. The intuition is that the impact of the DP noise is small when the sample size is large enough \cite{icdm:VuS09}. However, it is shown that certain tests, \eg, $\chi^2$-tests, perform poorly when used with the estimated statistic \cite{icml:RogersVLG16}, leading to much higher type-I error than the specified amount.
The second type of approaches \cite{jpc:UhleropSF13,jbi:YuFSU14,corr:WangLK15,icml:RogersVLG16} try to derive the asymptotic distribution of the estimated test statistic. Since this asymptotic distribution cannot be written analytically, Monte Carlo simulations or numerical approximations are used to calculate the p-value in step 2).
More recently, for $\chi^2$-tests, unit circle mechanism \cite{icml:KakizakiFS17} utilizes the geometrical property of the test statistics and achieves a sharp reduction on the type-II errors; and independently, new test statistics \cite{aistats:RogersK17} are proposed, so that their asymptotic distributions with DP noise injected are close to the asymptotics of the classical (non-private) tests.

To our best knowledge, our work is the first on statistical hypothesis tests to compare population means under LDP. One of our primary applications is A/B testing in software companies, so LDP is a proper privacy guarantee for each user without the need of trusting the data collector. LDP $\chi^2$-testing is studied in \cite{phd:Rogers17}, to test goodness of fit and independence for multinomial distributions. LDP hypothesis tests to distinguish between two specific distributions are studied in \cite{nips:KairouzOV14}. 

Another relevant line of works are about parameter estimations under LDP, including, \eg, mean/density estimations \cite{focs:DuchiJW13,corr:DuchiWJ16,nips:DingKY17}, and histogram estimations \cite{nips:DuchiWJ13,kairouz2016discrete,corr:WangHWNXYLQ16,edbt:0009WH16}. Communication and computation-efficient mechanisms are developed for histogram estimations over large domains to find heavy hitters \cite{stoc:BassilyS15,uss:WangBLJ17,corr:BassilyNST17}.
Industrial deployments of LDP techniques on this line enhance privacy using memorization \cite{ccs:ErlingssonPK14,popets:FantiPE16}.
%
% These techniques, however, cannot used directly to conduct hypothesis tests with bounded 
%
% Hypothesis testing v.s. LDP  \cite{corr:DuchiKR16}
%
How to find heavy hitters is also studied in a hybrid-privacy model with both LDP and DP users \cite{uss:AventKZHL17}.

\section{Preliminaries}

Each {\em user} has a private {\bf real-valued} {\em counter} $x \in \Sigma = [0, m]$ (to measure, \eg, app usage). Our approaches can be easily extended for general domains like $[-m,m]$, but we focus on $[0,m]$ for the simplicity of presentation. Let $[n] = \{1, 2, \ldots, n\}$ and $X = \{x_i\}_{i \in [n]}$ be a (sample) set of counters from $n$ users. We use $\mu_X = \sum_i x_i/n$ and $s^2_X = \sum_{i} (x_i - \mu_X)^2/(n-1)$ to denote the {\em sample mean} and {\em sample variance} of $X$, respectively. We use $\px$ (or $\pa$, $\pb$) to denote both a population and the distribution of this population, from which a sample $X$ (or $A$, $B$) is drawn.

% For example, collecting such counters to estimate a metric (\eg, app usage) and conduct A/B test at scales is very important to IT companies.

\stitle{Hypothesis testing to compare population means.}
%
% At its simplest, a randomized controlled experiment splits users randomly into two populations, gives them two  different experiences, a control and a treatment, and then tests for differences between the two groups in measurements of interest (clicks, usage, monetization, \etc). Through randomization, all the other factors between the control and the treatment group that may lead to differences in measurements are muted; therefore, any difference can be attributed to the treatment. Indeed, the different experiences, control and treatment, are only sent to sample users to minimize exposure (\ie, possible detriment effects) and costs, and the comparison of means of measurements can be performed on the samples using a $t$-test. 
%
% In this paper, we focus on comparing differences in a single integral or real measurement between two populations; though, our approach can easily be extended to multiple populations, multiple measurements, and other types of measurements (\eg, proportions). 
%
Let $\pa$ and $\pb$ be the distributions of counters in the {\em control} and the {\em treatment} populations, respectively. Let $\mu_\pa$ and ${\mu_\pb}$ be the {\em population means} ({\em expectations}). The {\em null hypothesis} $H_0$ is $\mu_\pa - \mu_\pb = d_0$, and the {\em alternative hypothesis} $H_1$ is, \eg, $\mu_\pa - \mu_\pb \neq d_0$.
We randomly pick $n_A$ and $n_B$ users from the populations $\pa$ and $\pb$, respectively, and let 
$A = \{a_i\}_{i \in [n_A]}$ and $B = \{b_i\}_{i \in [n_B]}$ be the corresponding samples.
% of counters from the distributions $\pa$ and $\pb$.

A test is an algorithm $\test$ that takes the two samples $A$ and $B$ and decides whether to reject or accept $H_0$ at a pre-specified {\em significance level} $\alpha$.
We require $\test$ to have {\em type-I error} at most $\alpha$, \ie, $\pr{\test(A, B; \alpha, H_0) = \text{reject} \mid H_0} \leq \alpha$, and {\em type-II error} $\pr{\test(A, B; \alpha, H_0) = \text{accept} \mid H_1} = \beta$ as small as possible. $1-\beta$ is called the {\em statistical power} of $\test$.
The probability is taken over the randomness from the data generation (of $A$ and $B$) and the possible randomness in $\test$.

A key step in a test is to compute the observed value of a test statistic from samples. To compare population means, it is usually a function of six parameters: sample means $\mu_A$ and $\mu_B$, sample variances $s^2_A$ and $s^2_B$, and sample sizes $n_A$ and $n_B$. In $t$-tests, we also need to obtain the {\em degrees of freedom}. For example, in Welch's $t$-test, they are, respectively
\begin{equation} \label{equ:tstatistic}
t \! = \! \frac{(\mu_A - \mu_B) - d_0}{\sqrt{s^2_A/n_A + s^2_B/n_B}} ~~\text{and}~~ df \! = \! \frac{(s_A^2/n_A + s_B^2/n_B)^2}{\frac{(s_A^2/n_A)^2}{n_A-1} + \frac{(s_B^2/n_B)^2}{n_B-1}}.
\end{equation}
%
% \begin{equation} \label{equ:df}
% df = \frac{(s_A^2/n_A + s_B^2/n_B)^2}{\frac{(s_A^2/n_A)^2}{n_A-1} + \frac{(s_B^2/n_B)^2}{n_B-1}}.
% \end{equation}

% In practice, experimenters compute $n_\pa = n_\pb$ needed to detect differences $\delta = \mu_\pa - \mu_\pb$ of specific sizes, given apriori knowledge about $\pa$ and $\pb$ with desired statistical confidence, typically $\alpha = .05$ and $\beta = .2$. 

\stitle{Local model of differential privacy (LDP).}
%
% {\em Differential privacy} (DP) was proposed by \cite{TCC06}. The most well-studied DP model is called the {\em global model}, where a trustable data collector first collects all users' data and then injects noise in the analysis step. However, in practice, users are unlikely to want unprivatized data to leave their devices, regardless of assurances about privacy later on. Hacking, leaks, and governmental intervention have all lead to user data privacy being compromised in the past. 
%
In the {\em local model of differential privacy} (LDP) \cite{nips:DuchiWJ13,stoc:BassilyS15}, also called randomized response model \cite{warner1965randomized}, $\gamma$-amplification \cite{evfimievski2003limiting}, or FRAPP \cite{agrawal2005framework}, private data from each user is randomized by an algorithm $\algo$ before being sent to data collector. 

% This is the privacy model adopted by both Apple, Google, Uber, and Microsoft.

\begin{definition}[Local model of differential privacy] \label{def:ldp}
A randomized algorithm $\algo: \Sigma \rightarrow \mathcal{Z}$ is $\epsilon$-locally differentially private ($\epsilon$-LDP) if for any pair of values $x, y \in \Sigma$ and any subset of output $S \subseteq \mathcal{Z}$, we have that \[\pr{\algo(x) \in S} \leq e^\epsilon\cdot \pr{\algo(y) \in S}.\]
\end{definition}

One interpretation of LDP is that no matter what output is released from $\algo$, it is approximately equally as likely to have come from one value $x \in \Sigma$ as any other. Unlike the DP model used in \cite{icml:RogersVLG16,aistats:RogersK17}, users do not need to trust the data collector in LDP.
%
% For alternate interpretations we refer the reader to \cite{wasserman2010statistical,focs:DuchiJW13}.

\stitle{Problem statement: LDP mean-comparison test.}
Each user in the control and the treatment has a counter. Two random samples of counters $A$ and $B$ are drawn from the control and the treatment distributions $\pa$ and $\pb$, respectively. With the {null hypothesis} $H_0$: $\mu_\pa - \mu_\pb = d_0$, we want to design a test, such that: i) each counter in $A$ and $B$ is collected from the user in an $\eps$-LDP way, ii) its type-I error $\leq$ significance level $\alpha$, and iii) type-II error is as small as possible.

% Note that the privacy guarantee required here is strictly stronger than the one provided by tests under the DP model with a trusted data collector, \eg, \cite{icml:RogersVLG16}, because each user has her data perturbed in an $\eps$-LDP way before sending it to the data collector, and no trustable data collector is needed.

% Under the LDP model, instead of the exact user's data $A$ and $B$, we collect $A^{\sf bit} = \algobit{\eps,m}(A) = \{\algobit{\eps,m}(a)\}_{a \in A}$ and $B^{\sf bit} = \algobit{\eps,m}(B)$, with $\eps$-LDP guaranteed for each user in $A$ and $B$.

% We would like to design $\test_\eps$ to minimize the effect of DP perturbation in testing and analyze the reduced statistical power. The LDP model is also more flexible in the sense that some users may require $\eps$-LDP while others can choose to send the exact counter values to the data collector. We will also introduce how to handle such ``mixed'' populations.

\stitle{Building block: 1-bit LDP data collection.}
We will utilize the following $\eps$-LDP mechanism $\algobit{\eps,m}$ from \cite{nips:DingKY17} to privatize each counter in samples $A$ and $B$.
For each user with a counter $x$, it generates a noisy bit ($0$ or $1$), independently, and sends to the data collector
\begin{equation}\label{equ:1bitmech}
\algobit{\eps,m}(x) =
\left\{
\begin{array}{ll}
\!\! 1,& \!\!\!\! \hbox{with probability} ~ \frac{1}{e^\epsilon+1}+\frac{x}{m}\cdot \frac{e^\epsilon-1}{e^\epsilon + 1}; \\
\!\! 0,& \!\!\!\! \hbox{otherwise}.
\end{array}
\right.
\end{equation}

\begin{proposition}\label{prop:1bitdp}
The mechanism $\algobit{\eps,m}$ is $\eps$-LDP.
\end{proposition}
%
% \begin{proof}
% In \appendixtext.
% \end{proof}

This mechanism can be interpreted as firstly rounding $x$ to a bit $1$ with probability ${x}/{m}$ or $0$ otherwise, and flipping the bit with probability $\frac{1}{e^\epsilon + 1}$.
It is communication-efficient (only one bit is sent) and can be seen as a simplification of the multidimensional mean-estimation mechanism in \cite{focs:DuchiJW13}.

$\algobit{\eps,m}$ can be used for mean estimation. Suppose there are $n$ users: $X = \{x_i\}_{i \in [n]}$. We collect $x'_i = \algobit{\eps,m}(x_i)$ from each user $i$. The mean $\mu_X = \sum_i x_i/n$ can be estimated from the noisy bits $X' = \algobit{\eps,m}(X) = \{x'_i\}_{i \in [n]}$ as:
\begin{equation}\label{equ:meanest}
\hat\mu_{\eps,m}(X') = \frac{m}{n}\sum_{i=1}^{n} \frac{x'_i \cdot (e^\eps+1)-1}{e^\eps-1}.
\end{equation}
It is shown in \cite{nips:DingKY17} that:
\begin{proposition}\label{prop:1bitmean}
The estimator $\hat\mu_{\eps,m}(X')$ is unbiased:\\i) $\ep{\hat\mu_{\eps,m}(X')} = \mu_X$, and ii) $\vr{\hat\mu_{\eps,m}(X')} = \bigoh{\frac{m^2}{n\eps^2}}$.
\end{proposition}

\section{Estimation-based LDP Test}

Given two samples $A = \{a_i\}_{i \in [n_A]}$ and $B = \{b_i\}_{i \in [n_B]}$, one straightforward starting point is to use $\algobit{\eps,m}$ to collect each counter, to preserve $\eps$-LDP for each user, and estimate parameters (sample means $\mu_A$ and $\mu_B$ and sample variances $s^2_A$ and $s^2_B$) in, \eg, \eqref{equ:tstatistic}, based on $\algobit{\eps,m}(A) = \{\algobit{\eps,m}(a_i)\}$ and $\algobit{\eps,m}(B) = \{\algobit{\eps,m}(b_i)\}$ to conduct a $t$-test.
%
% we use $\algobit{\eps,m}$ to collect counters in $A$ and $B$. The standard testing method, \eg, Welch's $t$-test in \eqref{equ:tstatistic} can be applied, as long as we can estimate the sample means $\mu_A$ and $\mu_B$, as well as the sample variances $s^2_A$ and $s^2_B$. 

We can obtain estimators to $\mu_A$ and $\mu_B$ using $\hat\mu_{\eps,m}$ \eqref{equ:meanest}. However, it is difficult to estimate the sample variances from the LDP data collection $\algobit{\eps,m}(A)$ and $\algobit{\eps,m}(B)$.
The intuition is as follows.  Consider two distributions: a counter from distribution $\px$ is always a constant $m/2$; a counter from $\py$ is $0$ with probability $1/2$, and $m$ otherwise. After applying $\algobit{\eps,m}$ on two samples $X$ and $Y$ from $\px$ and $\py$, respectively, the LDP samples $\algobit{\eps,m}(X)$ and $\algobit{\eps,m}(Y)$ follow the same distribution (cannot be distinguished from each other), but the gap between $s^2_X$ and $s^2_Y$ is $\bigomega{m^2}$ with high probability.
In general, we have a hardness result:

\begin{proposition} \label{prop:varhardness}
If each counter in $X = \{x_i\}$ is collected using $\algobit{\eps,m}$, any estimator $\hat s^2$ to $s^2_X$ based on $\algobit{\eps,m}(X) = \{\algobit{\eps,m}(x_i)\}$ has worst-case error $|\hat s^2 - s^2_X|$ at least $\bigomega{m^2}$.
\end{proposition}
\begin{proof}
In \appendixtext.
%
% Suppose $X = \{x_i\}$ is a random sample drawn from a distribution $\px$ such that $\pr{x_i = m/2} = 1$, and $Y = \{y_i\}$ is a sample drawn from $\py$ such that $\pr{y_i = 0} = 1/2$ and $\pr{y_i = m} = 1/2$. So $X$ has a sample variance $s^2_X = 0$, while with high probability, $Y$ has a sample variance at least $s^2_Y = \bigomega{m^2}$. However, the LDP bits $\algobit{\eps,m}(X) = \{\algobit{\eps,m}(x_i)\}$ and $\algobit{\eps,m}(Y)$ follows the same distribution:
% \begin{align*}
% \forall x_i \in X: &~ \pr{\algobit{\eps,m}(x_i) = 0} = \pr{\algobit{\eps,m}(x_i) = 1} = 1/2;
% \\
% \forall y_i \in Y: &~ \pr{\algobit{\eps,m}(y_i) = 0} = \pr{\algobit{\eps,m}(y_i) = 1} = 1/2.
% \end{align*}
% %
% The conclusion follows from the fact that $X$ and $Y$ cannot be distinguished from each other based on their $\eps$-LDP collections $\algobit{\eps,m}(X)$ and $\algobit{\eps,m}(Y)$, while their sample variances have a $\bigomega{m^2}$ gap with high probability.
\end{proof}

% Because of the above hardness result on estimating $s^2_A$ (note that $m^2$ is the largest possible sample variance if counter values are within $[0,m]$), we need to find a workaround to conduct hypothesis testing for $H_0$.

\stitle{LDP Data Collection for Variance Estimation.}
The above proposition essentially says that estimating the sample variance of a sample $X$ based on $\algobit{\eps,m}(X)$ leads to unbounded error, as a sample variance itself is bounded by $\bigoh{m^2}$.
In order to obtain a reasonable estimation for sample variances, we need to collect two LDP bits from each user.

Indeed, the sample variance $s^2_X$ can be rewritten as
\begin{equation} \label{equ:var}
s^2_X = \frac{1}{n-1}\sum_{i} \left(x_i - \mu_X\right)^2 = \frac{n}{n-1}\left(\mu_{X^2} - \mu_X^2\right),
\end{equation}
where $X^2$ is defined to be $\{x_i^2\}_{i \in [n]}$ and $\mu_{X^2} = \frac{1}{n}\sum_i x_i^2$.

The sequential composability \cite{sigmod:McSherry09} of DP also holds for LDP (considering a dataset with one user). We split the privacy budget into $\eps = \eps_1 + \eps_2$.
For each user with a counter $x_i$, the first bit to be collected is $x'_i = \algobit{\eps_1,m}(x_i)$, which can be also used to estimate mean $\mu_X$. The second bit is $x''_i = \algobit{\eps_2,m^2}(x_i^2)$, which will be used to estimate $\mu_{X^2}$ (the range of $x^2_i$ is $[0,m^2]$). Note that the two bits are collected with independent randomnesses.
From the sequential composability, we preserve $(\eps_1+\eps_2)$-LDP for each user.

After collecting $X' = \{x'_i\}_{i \in [n]}$ and $X'' = \{x''_i\}_{i \in [n]}$ from users in $X$,
the sample variance can be estimated as
\begin{equation}\label{equ:varest}
\hat s^2_{\eps_1, \eps_2, m}(X', X'') = \frac{n\left(\hat\mu_{\eps_2,m^2}(X'') - \hat\mu^2_{\eps_1,m}(X')\right)}{n-1},
\end{equation}
where $\hat\mu_{\eps_1,m}$ and $\hat\mu_{\eps_2,m^2}$ are defined as in \eqref{equ:meanest}.
We have the following result about the accuracy of $\hat s^2_{\eps_1, \eps_2, m}$.
%
% We can prove that $\hat s_{\eps_1, \eps_2, m}(X)$ is almost an unbiased estimator to $s^2_X$.
%
\begin{proposition}\label{prop:2bitvar}
$s^2_{\eps_1, \eps_2, m}(X', X'')$ is an estimator to $s^2_X$:
\\
i) $s_X^2 - \bigoh{\frac{m^2}{n\eps_1^2}} \leq \ep{\hat s^2_{\eps_1, \eps_2, m}(X', X'')} \leq s_X^2$, and
\\
ii) $\vr{\hat s^2_{\eps_1, \eps_2, m}(X', X'')} = \bigoh{\frac{m^4}{n^2\eps_1^4} + \frac{m^4}{n\eps_2^2}}$.
\end{proposition}
\begin{proof}
In \appendixtext.
\end{proof}

% The above result is of independent interest when we just want to estimate variance in an LDP way. Both the bias and the variance of $\hat s^2_{\eps_1, \eps_2, m}$ diminish as the sample size $n$ grows. As we require that $\eps_1 + \eps_2 = \eps$, there is a tradeoff between the bias (smaller $\eps_1$ incurs larger bias) and the estimation variance (smaller $\eps_2$ incurs larger variance).

\stitle{Using Mean/Variance Estimation in Tests.} The test $\testest_{\eps_1, \eps_2}$ uses the above mechanism to collect data and estimate $\mu_A$, $\mu_B$, $s^2_A$, and $s^2_B$ under LDP (consider $X = A, B$), and put the estimates back into \eqref{equ:tstatistic} to calculate $t$ and $df$ in order to conduct $t$-tests. Refer to \appendixtext for detailed description. There is no theoretical guarantee on the testing errors, but from the above discussion, we have:
\begin{theorem}\label{thm:esttest}
$\testest_{\eps_1, \eps_2}$ satisfies $(\eps_1+\eps_2)$-LDP.
\end{theorem}
%
% \begin{proof}
% It follows from \cprop\ref{prop:1bitdp}, how $A'$, $A''$, $B'$ and $B''$ are collected, and the sequential composability of DP.
% \end{proof}

\section{Transformation-based LDP Test}

In this section, we give a LDP testing algorithm with guaranteed significance and power. The main idea is that, if a counter $x$ follows some (unknown) population distribution with a (unknown) mean $\mu$, the LDP bit $\algobit{\eps,m}(x)$ follows a Bernoulli distribution with the mean determined by $\mu$ and $\eps$. So in order to compare population means, we can conduct tests directly on the LDP bits and compare Bernoulli means. 

The following proposition firstly gives the relationship between the original population distribution and the resulting Bernoulli distribution on the outputs of $\algobit{\eps,m}$.

\begin{proposition} \label{prop:bindistribution}
If a counter $x \in [0, m]$ follows a distribution $\px$ with mean $\mu_\px$, the LDP bit $x^{\sf bin}=\algobit{\eps,m}(x)$ (as in \eqref{equ:1bitmech}) follows a Bernoulli distribution with the mean
\begin{equation} \label{equ:bindistribution}
p_\px = \pr{x^{\sf bin} = 1} = \frac{\mu_\px}{m} \cdot \frac{e^\eps-1}{e^\eps+1} + \frac{1}{e^\eps+1}.
\end{equation}
\end{proposition}
\begin{proof}
Let $f$ be the PDF of $\px$. We have
\begin{align*}
& \pr{x^{\sf bin} = 1} = \int_0^m \pr{x^{\sf bin} = 1 \mid x} f(x) {d}x
\\
= & \int_0^m \left(\frac{1}{e^\epsilon+1}+\frac{x}{m}\cdot \frac{e^\epsilon - 1}{e^\epsilon + 1}\right) f(x) {d}x
\\
= & \frac{1}{e^\eps+1} \int_0^m f(x)dx + \frac{1}{m}\cdot \frac{e^\epsilon - 1}{e^\epsilon + 1} \int_0^m x f(x) {d}x
\\
= & \frac{1}{e^\eps+1} \cdot 1 + \frac{1}{m}\cdot \frac{e^\epsilon - 1}{e^\epsilon + 1} \cdot \ep{x}
\\
= & \frac{\mu_\px}{m} \cdot \frac{e^\eps-1}{e^\eps+1} + \frac{1}{e^\eps+1}.
\end{align*}
Therefore, $x^{\sf bin}$ follows a Bernoulli distribution.
\end{proof}

\begin{algorithm}[t]
{\bf Input:} Two samples $A = \{a_i\}_{i \in [n_A]}$ and $B = \{b_i\}_{i \in [n_B]}$.
\\
{\bf Null hypothesis $H_0$:} $\mu_\pa - \mu_\pb = d_0$.
\\
{\bf Parameters:} privacy budget $\eps$ and significance level $\alpha$. \\ \vspace{-0.4cm}
\begin{algorithmic}[1]
\STATE For users $i = 1$ to $n_A$ do
\STATE ~~~ Encode $a^{\sf bin}_i = \algobit{\eps,m}(a_i)$ and send $a^{\sf bin}_i$ to the server.
\STATE For users $i = 1$ to $n_B$ do
\STATE ~~~ Encode $b^{\sf bin}_i = \algobit{\eps,m}(b_i)$ and send $b^{\sf bin}_i$ to the server.
\STATE Server receives $A^{\sf bin} = \{a^{\sf bin}_i\}$ and $B^{\sf bin} = \{b^{\sf bin}_i\}$.
\STATE Let the transformed null hypothesis be
\[H_0^{\sf bin}: p_\pa - p_\pb = \frac{d_0}{m} \cdot \frac{e^\eps-1}{e^\eps+1},\]
where $p_\pa$ ($p_\pb$) is the distribution mean of $A^{\sf bin}$ ($B^{\sf bin}$).
\STATE Conduct a $t$-test with null hypothesis $H_0^{\sf bin}$ on $A^{\sf bin}$ and $B^{\sf bin}$ at significance level $\alpha$: accept (or reject) $H_0$ if and only if $H_0^{\sf bin}$ is accepted (or rejected).
\end{algorithmic}
\caption{$\testbin_{\eps}$: Transformation-based LDP Test}
\label{alg:bittest}
\end{algorithm}

The test process $\testbin_{\eps}$ is described in \calg\ref{alg:bittest}. Two samples, $A = \{a_i\}_{i \in [n_A]}$ and $B = \{b_i\}_{i \in [n_B]}$, are drawn from two distributions $\pa$ and $\pb$ with means $\mu_\pa$ and $\mu_\pb$. After applying the mechanism $\algobit{\eps,m}$, the LDP bits $A^{\sf bin} = \algobit{\eps,m}(A) = \{\algobit{\eps,m}(a_i)\}_{i \in [n_A]}$ and $B^{\sf bin}$ $=$ $\algobit{\eps,m}(B)$ collected (lines~1-5) are two samples from Bernoulli distributions with means $p_\pa$ and $p_\pb$. With $\mu_\pa - \mu_\pb = d_0$, from \eqref{equ:bindistribution}, we have $p_\pa - p_\pb = d_0^{\sf bin} = ({d_0}/{m}) \cdot ((e^\eps-1)/(e^\eps+1))$.
An important observation here is that the relative order between $\mu_\pa$ and $\mu_\pb$ is the same as the one between $p_\pa$ and $p_\pb$, \ie, $\mu_\pa - \mu_\pb \geq d_0$ $\Leftrightarrow$ $p_\pa - p_\pb \geq d_0^{\sf bin}$. Therefore, we can conduct a test on $A^{\sf bin}$ and $B^{\sf bin}$ to compare $p_\pa$ and $p_\pb$ with a null hypothesis $H_0^{\sf bin}: p_\pa - p_\pb = d_0^{\sf bin}$ (line~6), in order to compare $\mu_\pa$ and $\mu_\pb$ and reject or accept $H_0$ (line~7).

Indeed, $\testbin_\eps$ preserves $\eps$-LDP for each user from \cprop\ref{prop:1bitdp} and the way how $A^{\sf bin}$ and $B^{\sf bin}$ are collected.
\begin{theorem} \label{thm:bittest:privacy}
$\testbin_{\eps}$ satisfies $\eps$-LDP.
\end{theorem}
%
% \begin{proof}
% The privacy guarantee $\eps$-LDP for each user follows from \cprop\ref{prop:1bitdp} and how $A^{\sf bin}$ and $B^{\sf bin}$ are collected.
% \end{proof}

We do need a larger sample size in $\testbin_{\eps}$ to get a satisfactory statistical power than the size needed in a non-private $t$-test on the real values in $A$ and $B$.
Intuitively, for a fixed gap between population means $\mu_\pa$ and $\mu_\pb$, the gap between $p_\pa$ and $p_\pb$ is smaller if the domains size $m$ is larger or $\eps$ is smaller; note that the smaller the gap between $p_\pa$ and $p_\pb$ is, the harder for the test $\testbin_{\eps}$ to draw a significant conclusion.
Lower bounds of the power of $\testbin_{\eps}$ (or sample sizes needed) are derived in \cthm\ref{thm:bittest:power} at a significance level $\alpha$.

% Intuitively, the sample size ($n_A + n_B$) has to be large enough to cancel out the noise in LDP bits.
%
% As a numeric example, let the domain size be $m = 100$, the difference in means $\theta = 5$, the significance level $\alpha = 0.05$, the privacy parameter $\eps = 1$, and the sample size $n = 160,000$ -- the statistical power is at least $99.86\%$.

\begin{theorem} \label{thm:bittest:power}
$\testbin_{\eps}$ (\calg\ref{alg:bittest}) has a significance level $\alpha$, \ie, type-I error $\leq \alpha$.
Suppose the alternative $H_1$: $\mu_\pa - \mu_\pb > d_0$ is true with $(\mu_\pa - \mu_\pb) - d_0 = \theta$. The statistical power ($1 -$ type-II error) of $\testbin_\eps$, denoted by $P(\theta)$, is
\begin{align}
& \!\!\! P(\theta) \! \triangleq \! \pr{\testbin_\eps(A, B) = \text{\rm reject} \mid (\mu_\pa - \mu_\pb) - d_0 = \theta} \! \geq \! \nonumber
\\
& \!\! 1 - \exp\left(-{\left(\frac{\theta(e^\eps-1)}{m(e^\eps+1)} \sqrt{\frac{2n_An_B}{n_A+n_B}}- \sqrt{\ln\frac{1}{\alpha}}\right)^2}\right), \label{equ:sp0}
\end{align}
if the samples are large enough: $n_A + n_B = \bigomega{\frac{m^2}{\theta^2 \eps^2}\ln\frac{1}{\alpha}}$.

If we use Normal distributions to approximate Binomial distributions and Student's $t$-distributions (under the condition that $n_A$ and $n_B$ are large enough), we have:
\begin{align}
& P(\theta) \geq 1 - F\left(F^{-1}\left(1-{\alpha}\right) - \frac{p_\theta}{\hat\sigma_{\bf A + B}}\right) \geq \label{equ:spnormal1}
\\
& 1 - F\left(F^{-1}\left(1-{\alpha}\right) - p_\theta \cdot {\sqrt{\frac{4(n_A-1)(n_B-1)}{n_A+n_B-2}}}\right) \label{equ:spnormal2}
\end{align}
% \begin{align*}
% P(\theta) & = 1 - F\left(F^{-1}\left(1-{\alpha}\right) - \frac{p_\theta}{\sigma_{\bf A + B}}\right)
% \\
% & \geq 1 - F\left(F^{-1}\left(1-{\alpha}\right) - p_\theta \cdot {\sqrt{\frac{4n_An_B}{n_A+n_B}}}\right),
% \end{align*}
where $p_\theta = \frac{\theta}{m} \cdot \frac{e^\eps-1}{e^\eps+1}$,
%
% $\sigma_{\bf A + B}$ is the sample standard deviation
%
$F(\cdot)$ is the CDF of the Normal distribution $\normal(0,1)$, and the sample variance $\hat\sigma_{\bf A+B} = $
\begin{align} 
& \!\! \sqrt{\frac{\one_{A^{\sf bin}}/n_A - \one^2_{A^{\sf bin}}/n^2_A}{n_A-1} + \frac{\one_{B^{\sf bin}}/n_B - \one^2_{B^{\sf bin}}/n^2_B}{n_B-1}}, ~\text{where} \nonumber
\\
& \!\! \text{$\one_X = |\{x \in X \mid x = 1\}|$ is the number of $1$'s in $X$.} \label{equ:samplevarbin}
\end{align}

In particular, if $n_A = n_B = n$ and we require that the statistic power $P(\theta) \geq 1-\beta$, it suffices to have
\begin{equation} \label{equ:esnormal}
n = \left(F^{-1}\left(1-{\alpha}\right) - F^{-1}(\beta)\right)^2 \cdot \frac{1}{2p_\theta^2} + 1.
\end{equation}
\end{theorem}
\begin{proof}
%
% From the above proposition, if $A = \{a_i\}$ and $B = \{b_i\}$ are two samples drawn from two distributions with means $\mu_\pa$ and $\mu_\pb$, the corresponding LDP bits $A^{\sf bit} = \algobit{\eps,m}(A) = \{\algobit{\eps,m}(a)\}_{a \in A}$ and $B^{\sf bit} = \algobit{\eps,m}(B)$ are two samples drawn from Bernoulli distributions with means $p(\mu_\pa)$ and $p(\mu_\pb)$, respectively. Therefore, we accept (reject) the null hypothesis $H_0:~ \mu_\pa = \mu_\pb$, if and only if $H^{\sf bit}_0:~ p(\mu_\pa) = p(\mu_\pb)$ is accepted (rejected) in the binomial test conducted on $A^{\sf bit}$ and $B^{\sf bit}$. Indeed, we lose the statistical power because of noise injected into LDP bits. We will give an lower bound estimation of the statistical power in the following part.
%
% It is important to note that the statistical power we want to estimate is for the test $\testbin_\eps$ but not for the binomial test (a subroutine of $\testbin_\eps$). More formally, we want to estimate:
% \begin{align}
% \text{statistical power} & \triangleq \pr{\testbin_\eps(A, B; \alpha, H_0) = \hbox{reject} \mid H_1} \nonumber
% \\
% & = \pr{ {\rm reject}~H^{\sf bit}_0 \mid \mu_\pa \neq \mu_\pb}.
% \end{align}
%
% We first focus on two-tailed tests and two samples with the same size $n$. For one-tailed tests, the conclusion differs only by a factor of two. Let $|\mu_\pa - \mu_\pb| = \theta$, and we want to estimate the statistical power at a significance level $\alpha$:
% \[
% P(\theta, m, \alpha, n, \eps) \triangleq \pr{ {\rm reject}~H^{\sf bit}_0 \mid |\mu_\pa - \mu_\pb| = \theta}.
% \]
%
Let's focus on the case $d_0 = 0$. The proof can be easily generalized for $d_0 > 0$ by adding a constant.

Consider a test with a null hypothesis $H_0^{\sf bin}$ (in line~7 of \calg\ref{alg:bittest}) using the following test statistic:
\[
z(A^{\sf bin}, B^{\sf bin}) = \frac{1}{n_A}\one_{A^{\sf bin}} - \frac{1}{n_B}\one_{B^{\sf bin}} 
\]
($\one_X$ is defined in \eqref{equ:samplevarbin}).
% $\one_X = |\{x \in X \mid x = 1\}|$ is the number of ones in a set $X$ of bits.
Using the linearity of expectation, we have $\ep{z(A^{\sf bin}, B^{\sf bin})} = p_\pa - p_\pb$ ($p_\px$ is defined in \eqref{equ:bindistribution}).

The proof of \eqref{equ:sp0} is in \appendixtext, using a weaker test and McDiarmid's inequality.

We now focus on \eqref{equ:spnormal1}-\eqref{equ:spnormal2}.
%
% We can use Normal approximations to derive the rest lower bounds of $P(\theta)$. 
%
Let's first state the Normal approximation: {\em a Binomial distribution $\binomial(n,p)$ with $n$ trials and success probability $p$ can be approximated by a normal distribution $\normal(np, np(1-p))$, if $n$ is large enough.}
% \footnote{The condition is that, \eg, $np \pm 3\sqrt{np(1-p)} \in (0,n)$.}

From \cprop\ref{prop:bindistribution}, we have $\one_{A^{\sf bin}} \sim \binomial(n_A, p_\pa)$ and $\one_{B^{\sf bin}} \sim \binomial(n_B, p_\pb)$. Using Normal approximations to $\one_{A^{\sf bin}}$ and $\one_{B^{\sf bin}}$, under $H^{\sf bit}_0$ (when $d_0 = 0$), we have
\[
\frac{z(A^{\sf bin}, B^{\sf bin})}{\sigma_{\bf A + B}} \sim \normal(0,1),
\]
where $\sigma_{\bf A + B} = \sqrt{p_\pa(1-p_\pa)/n_A+p_\pb(1-p_\pb)/n_B}$.

In the test (line~7), we use $\hat\sigma_{\bf A + B}$ in \eqref{equ:samplevarbin}
%
% \begin{equation} \label{equ:samplevarbin}
% \hat\sigma_{\bf A + B} = \sqrt{\frac{\one_{A^{\sf bin}}/n_A - \one^2_{A^{\sf bin}}/n^2_A}{n_A-1} + \frac{\one_{B^{\sf bin}}/n_B - \one^2_{B^{\sf bin}}/n^2_B}{n_B-1}}
% \end{equation}
%
to approximate $\sigma_{\bf A + B}$. When $n_A$ and $n_B$ are large enough,
\[
\frac{z(A^{\sf bin}, B^{\sf bin})}{\hat\sigma_{\bf A + B}} \sim \normal(0,1),
\]
from the Normal approximation to Student's $t$-distribution.

At a significance level of $\alpha$, we need to find an rejection threshold $z_0$ of $z(A^{\sf bin}, B^{\sf bin})$, \st, under $H_0^{\sf bin}$, $\pr{\text{\rm reject}}$ $=$ $\pr{z(A^{\sf bin}, B^{\sf bin}) \geq z_0} \leq \alpha$.
Therefore, based on the CDF of the Normal distribution $\normal(0,1)$, we reject $H^{\bf bit}_0$ iff
\[
z(A^{\sf bin}, B^{\sf bin}) \geq z_0 = F^{-1}(1-\alpha) \cdot \hat\sigma_{\bf A + B}.
\]

Now let's estimate the statistical power under the alternative hypothesis with $\mu_\pa - \mu_\pb = \theta$, or equivalently,
\[
p_\pa - p_\pb = \frac{\theta}{m} \cdot \frac{e^\eps-1}{e^\eps+1} \triangleq p_\theta.
\]
Under the above condition, we have
\[
\frac{z(A^{\sf bin}, B^{\sf bin}) - p_\theta}{\hat\sigma_{\bf A + B}} \sim \normal(0,1).
\]
And thus the statistical power $P(\theta) =$
\begin{align}
& = \pr{\frac{z(A^{\sf bin}, B^{\sf bin})}{\hat\sigma_{\bf A + B}} \geq F^{-1}\left(1-{\alpha}\right) ~\middle|~ p_\pa - p_\pb = p_\theta} \nonumber
\\
& = \pr{\frac{z(A^{\sf bin}, B^{\sf bin})-p_\theta}{\hat\sigma_{\bf A+B}} \geq F^{-1}\left(1-{\alpha}\right)-\frac{p_\theta}{\hat\sigma_{\bf A+B}}} \nonumber
\\
& = 1 - F\left(F^{-1}\left(1-{\alpha}\right) - \frac{p_\theta}{\hat\sigma_{\bf A + B}}\right) \nonumber
\\
& \geq 1 - F\left(F^{-1}\left(1-{\alpha}\right) - p_\theta \cdot {\sqrt{\frac{4(n_A-1)(n_B-1)}{n_A+n_B-2}}}\right). \nonumber
\end{align}
The sample size lower bound in \eqref{equ:esnormal} is directly from \eqref{equ:spnormal2}.
\end{proof}

\stitle{How to use \cthm\ref{thm:bittest:power}.} All three of \eqref{equ:sp0}-\eqref{equ:spnormal2} can be used to estimate the lower bound of statistical power, and the largest one can be picked. \eqref{equ:spnormal1} is likely to be the tightest one, but it needs the sample variance $\hat\sigma_{\bf A+B}$ \eqref{equ:samplevarbin} in the LDP samples $A^{\sf bin}$ and $B^{\sf bin}$, while the other two only need the sample sizes. Before we draw samples from populations, \eqref{equ:esnormal} can be used to estimate the sample sizes needed. As will be verified in the experiments, the estimated sample sizes are sufficient to guarantee the required statistical power.
It is interesting to note that the type-II error of $\testbin_{\eps}$ has a dominated term $\bigoh{-\exp(\sqrt{n})}$ similar to the one of unit circle mechanism \cite{icml:KakizakiFS17} for a different test ($\chi^2$-test) under DP. The additional term $\frac{1}{m}$ is from LDP.

\stitle{About Laplace mechanism.} $\testbin_{\eps}$ can be adapted if each user's counter is collected using the Laplace-perturbation mechanism \cite{focs:DuchiJW13}. However, sending Laplace-perturbed counters is costly and we cannot expect better statistical power from it.

\subsection{Extensions for Hybrid Privacy Requirements}
{\bf Hybrid privacy model.} The transformation-based LDP test $\testbin_{\eps}$ can be extended for population with {\em hybrid privacy requirements}:
more formally, {\em in a random sample $A = \{a_i\}$ (or $B = \{b_i\}$) drawn from the distribution $\pa$ (or $\pb$), some users require $\eps$-LDP, while the others do not.}

\stitle{Rescaling LDP bits.}
Indeed, for users who do not require $\eps$-LDP, we can simply send their exact counter $a_j$ (or $b_j$) to the server. The question is, for those who require $\eps$-LDP, \eg, $a_i$, how to combine their LDP bits $\algobit{\eps,m}(a_i) \in \{0, 1\}$ with exact counters $a_j \in [0, m]$ to conduct hypothesis tests.

\begin{algorithm}[t]
{\bf Input:} Two samples $A = \{a_i\}_{i \in [n_A]}$ and $B = \{b_i\}_{i \in [n_B]}$.
\\
{\bf Null hypothesis $H_0$:} $\mu_\pa - \mu_\pb = d_0$.
\\
{\bf Parameters:} privacy budget $\eps$ and significance level $\alpha$. \\ \vspace{-0.4cm}
\begin{algorithmic}[1]
\STATE For users $i = 1$ to $n_A$ do
\STATE ~~~ If $a_i$ requires $\eps$-LDP then:
\STATE ~~~ ~~~ If $\algobit{\eps,m}(a_i) = 0$ then: $a^{\sf mix}_i = - {m}/{(e^\eps-1)}$;
\STATE ~~~ ~~~ Else: $a^{\sf mix}_i = {m e^\eps}/{(e^\eps-1)}$.
\STATE ~~~ Else: $a^{\sf mix}_i = a_i$.
\STATE ~~~ Send $a^{\sf mix}_i$ to the server.
\STATE For users $i = 1$ to $n_B$ do
\STATE ~~~ If $b_i$ requires $\eps$-LDP then:
\STATE ~~~ ~~~ If $\algobit{\eps,m}(b_i) = 0$ then: $b^{\sf mix}_i = - {m}/{(e^\eps-1)}$;
\STATE ~~~ ~~~ Else: $b^{\sf mix}_i = {m e^\eps}/{(e^\eps-1)}$.
\STATE ~~~ Else: $b^{\sf mix}_i = b_i$.
\STATE ~~~ Send $b^{\sf mix}_i$ to the server.
\STATE Server receives $A^{\sf mix} = \{a^{\sf mix}_i\}$ and $B^{\sf mix} = \{b^{\sf mix}_i\}$.
\STATE Conduct a $t$-test with the null hypothesis $H_0^{\sf mix}: \mu_\pa^{\sf mix} - \mu_\pb^{\sf mix} = d_0$ on $A^{\sf mix}$ and $B^{\sf mix}$ at significance level $\alpha$: accept (or reject) $H_0$ iff $H_0^{\sf mix}$ is accepted (or rejected).
\end{algorithmic}
\caption{$\testmix_{\eps}$: For Hybrid Privacy Requirements}
\label{alg:bittest:mix}
\end{algorithm}

The proposed test $\testmix_{\eps}$ in \calg\ref{alg:bittest:mix} ``re-scales'' LDP bits $\algobit{\eps,m}(a_i)$ to form a mixed sample $A^{\sf mix}$ together with the exact counters. For a user $a_i$ who requires $\eps$-LDP, if $\algobit{\eps,m}(a_i) = 0$, $a^{\sf mix}_i = - {m}/{(e^\eps-1)}$ is sent (line~3), and if $\algobit{\eps,m}(a_i) = 1$, $a^{\sf mix}_i = {m e^\eps}/{(e^\eps-1)}$ is sent (line~4). If a user $a_i$ does not requires $\eps$-LDP, simply send  $a^{\sf mix}_i = a_i$ (line~5). The same process is applied for sample $B$ (lines~7-12). This process can be easily extended if different users require different values of the privacy parameter $\eps$.

We can show that $A^{\sf mix} = \{a^{\sf mix}_i\}$ and $B^{\sf mix} = \{b^{\sf mix}_i\}$ received by the server (line~13) can be considered as being drawn from distributions $\pa^{\sf mix}$ and $\pb^{\sf mix}$, with means $\mu_\pa^{\sf mix} = \mu_\pa$ and $\mu_\pb^{\sf mix} = \mu_\pb$, respectively, but higher variances.
\begin{proposition} \label{prop:mixdistribution}
If a counter $x \in [0, m]$ follows a distribution $\px$ with mean $\mu_\px$, $x^{\sf mix}$ (derived as $a_i^{\sf mix}$ or $b_i^{\sf mix}$ in \calg\ref{alg:bittest:mix}) follows a distribution $\px^{\sf mix}$ with mean $\mu_\px^{\sf mix} =  {\mu_\px}$.
%
% \begin{equation} \label{equ:mixdistribution}
% \mu_\px^{\sf mix} = \ep{x^{\sf mix}} = {\mu_\px}.
% \end{equation}

Let $\sigma^2_\px$ be the variance of $\px$. If each user in $\px$ requires $\eps$-LDP with probability $r$, the variance of $\px^{\sf mix}$ is
\begin{equation} \label{equ:mixdistribution:var}
\sigma^2_{\px^{\sf mix}} = \sigma^2_\px \cdot (1-r) + \bigoh{m^2/\eps^2} \cdot r.
\end{equation}
\end{proposition}
\begin{proof}
$\mu_\px^{\sf mix} =  {\mu_\px}$ is from \eqref{equ:bindistribution} and the fact that if $x$ requires LDP, $x^{\sf mix} = \algobit{\eps,m}(x) \cdot m \cdot (e^\eps+1)/(e^\eps-1) - {m}/{(e^\eps-1)}$.

The variance of $\px^{\sf mix}$ can be then calculated as follows:
\begin{align*}
& \sigma^2_{\px^{\sf mix}} = \ep{(x^{\sf mix} - \ep{x^{\sf mix}})^2}
\\
= & \ep{(x^{\sf mix} - \ep{x^{\sf mix}})^2 \mid \text{$x$ does not require LDP}} \!\! \cdot \! (1-r)
\\
+ & \ep{(x^{\sf mix} - \ep{x^{\sf mix}})^2 \mid \text{$x$ requires LDP}} \!\! \cdot \! r
\\
% = & \sigma^2_\px \cdot (1-r) + \bigoh{{m^2}/{\eps^2}} \cdot \vr{\algobit{\eps,m}(x)} \cdot r,
= & \sigma^2_\px \cdot (1-r) + \frac{m^2(e^\eps+1)^2}{(e^\eps-1)^2} \cdot \vr{\algobit{\eps,m}(x)} \cdot r,
\end{align*}
where $\vr{\algobit{\eps,m}(x)} = \bigoh{1}$ follows from \eqref{equ:1bitmech}.
\end{proof}

Since the distributions $\pa$ and $\pa^{\sf mix}$ (or, $\pb$ and $\pb^{\sf mix}$) have the same mean, we can conduct a $t$-test with the null hypothesis $H_0^{\sf mix}: \mu_\pa^{\sf mix} - \mu_\pb^{\sf mix} = d_0$ on $A^{\sf mix}$ and $B^{\sf mix}$ in order to accept or reject $H_0$ on $A$ and $B$ (line~14) -- this is because $H_0$ is a necessary and sufficient condition of $H_0^{\sf mix}$. For the same reason, the significance and the power of $\testmix_{\eps}$ are guaranteed to be the same as those of the $t$-test on line~14.
\begin{theorem} \label{thm:bittest:mix:power}
$\testmix_{\eps}$ (\calg\ref{alg:bittest:mix}) satisfied the hybrid privacy model. It has a significance level $\alpha$, \ie, type-I error $\leq \alpha$. Its power is the same as the power of the $t$-test on line~14.
\end{theorem}
\begin{proof}
The significance/power guarantee is from the above discussion.
The privacy guarantee for each user follows from \cprop\ref{prop:1bitdp} and how $A^{\sf mix}$ and $B^{\sf mix}$ are collected. 
\end{proof}

We do not give a closed-form lower bound of $\testmix_{\eps}$'s statistical power. Since it is equal to the power of the test conducted on line~14 of \calg\ref{alg:bittest:mix}, it depends on the variances of $\pa^{\sf mix}$ and $\pb^{\sf mix}$, which are determined by $\sigma^2_\pa$, $\sigma^2_\pb$, $r$ (the fraction of users requiring LDP), and $m$ as in \eqref{equ:mixdistribution:var}. And since $\sigma^2_\pa$ and $\sigma^2_\pb$ are usually much smaller than $m^2$, $r$ has a significant impact on the power: the smaller $r$ is, the smaller the variances $\sigma^2_{\pa^{\sf mix}}$ and $\sigma^2_{\pb^{\sf mix}}$ are, and the larger the power is. This property will be verified in our experiments.

\stitle{More general hybrid privacy models.} $\testmix_{\eps}$ can be easily extended for the model where each user $i$ requires a different privacy budget $\eps_i$ (just replacing $\eps$ with $\eps_i$ in lines~3-4 and 9-10 of \calg\ref{alg:bittest:mix}). In a different model considered for the heavy-hitter problem \cite{uss:AventKZHL17}, some users require LDP while others only require DP on the testing output: how to conduct effective  tests in this model remains open.

% \section{Discussion}
%
% How to spend sampling budget more smartly

\begin{figure}[t]
\centering
\includegraphics[width=8.3cm]{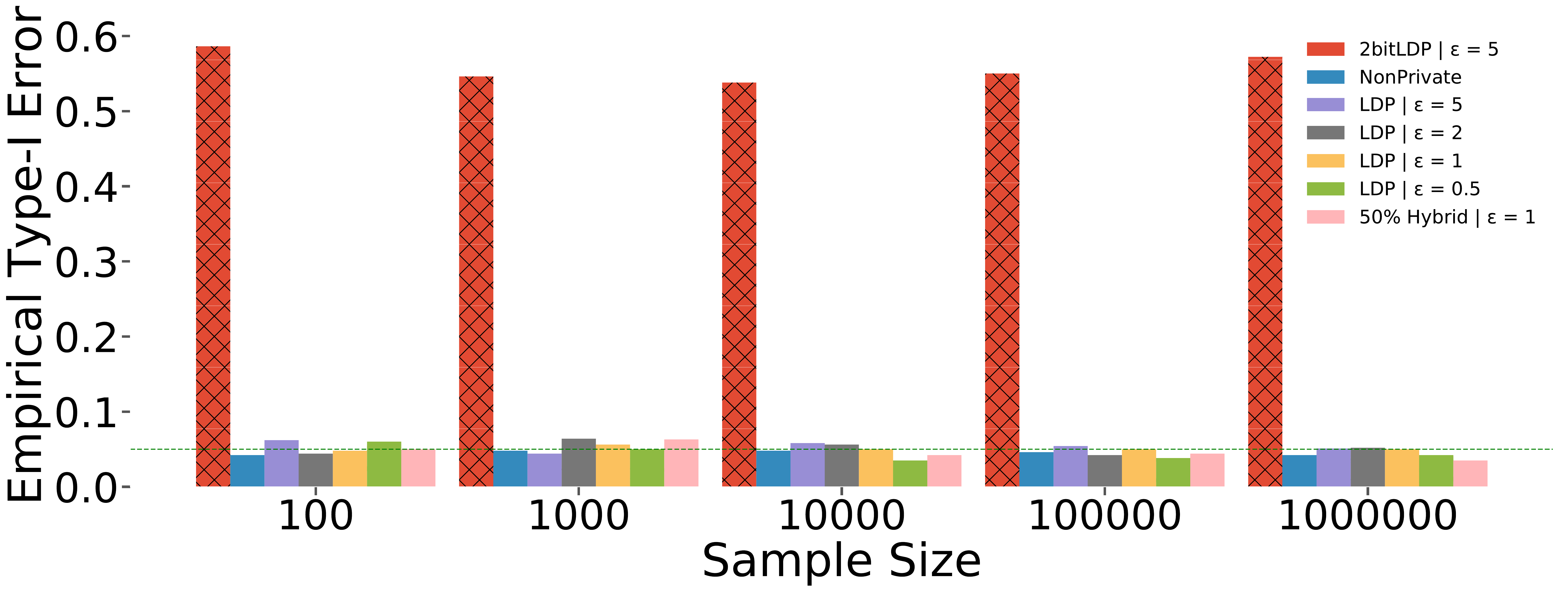}
%\vspace{-0.4cm}
\caption{Empirical type-I error of different approaches}
%\vspace{-0.4cm}
\label{fig:exp:type1}
\end{figure}

\begin{figure*}[t]\centering
\subfigure[Gap $\theta = 60$]{
\label{fig:exp:type2:60}
\includegraphics[height=3.9cm]{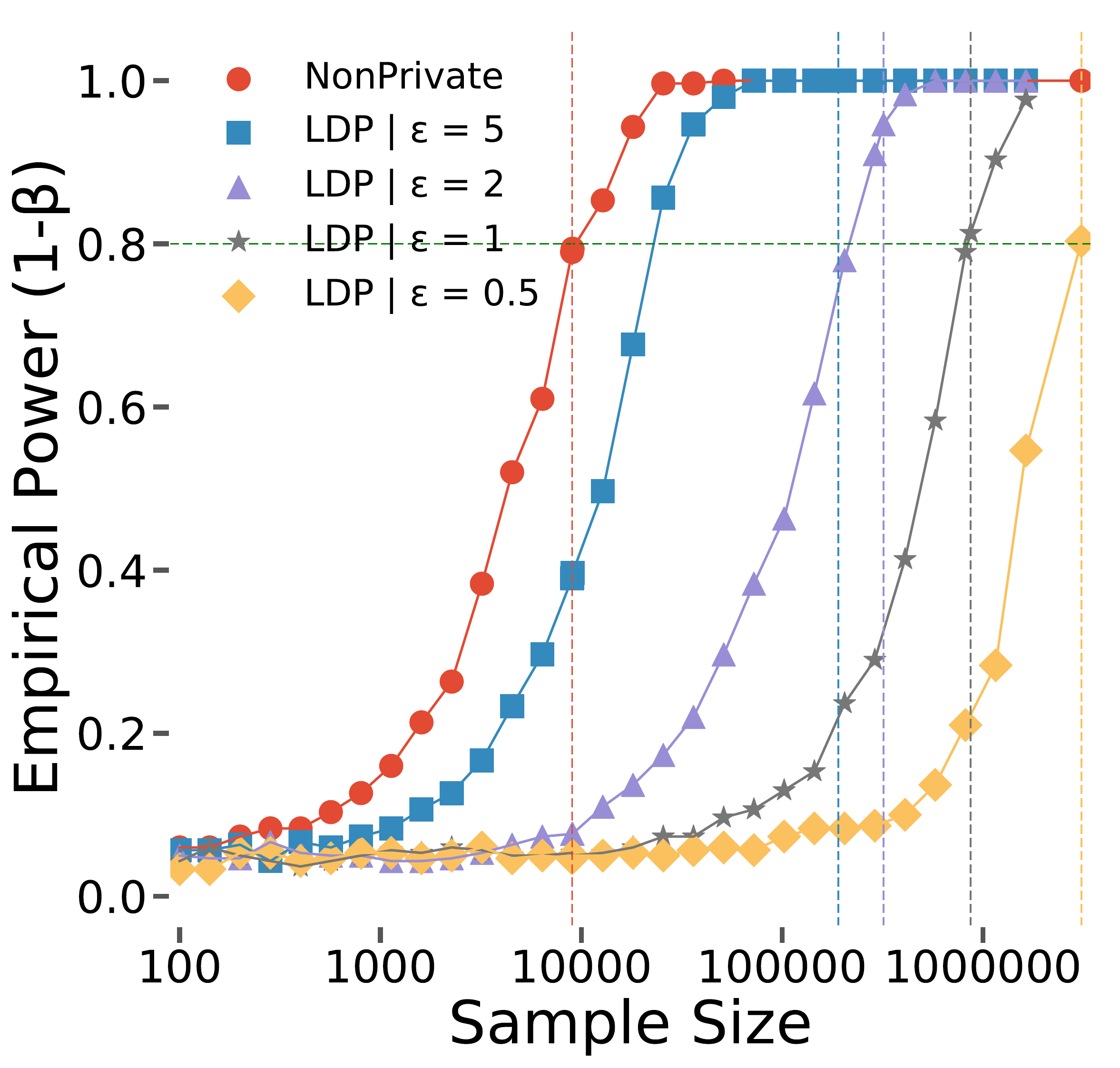}
}
\subfigure[Gap $\theta = 120$]{
\label{fig:exp:type2:120}
\includegraphics[height=3.9cm]{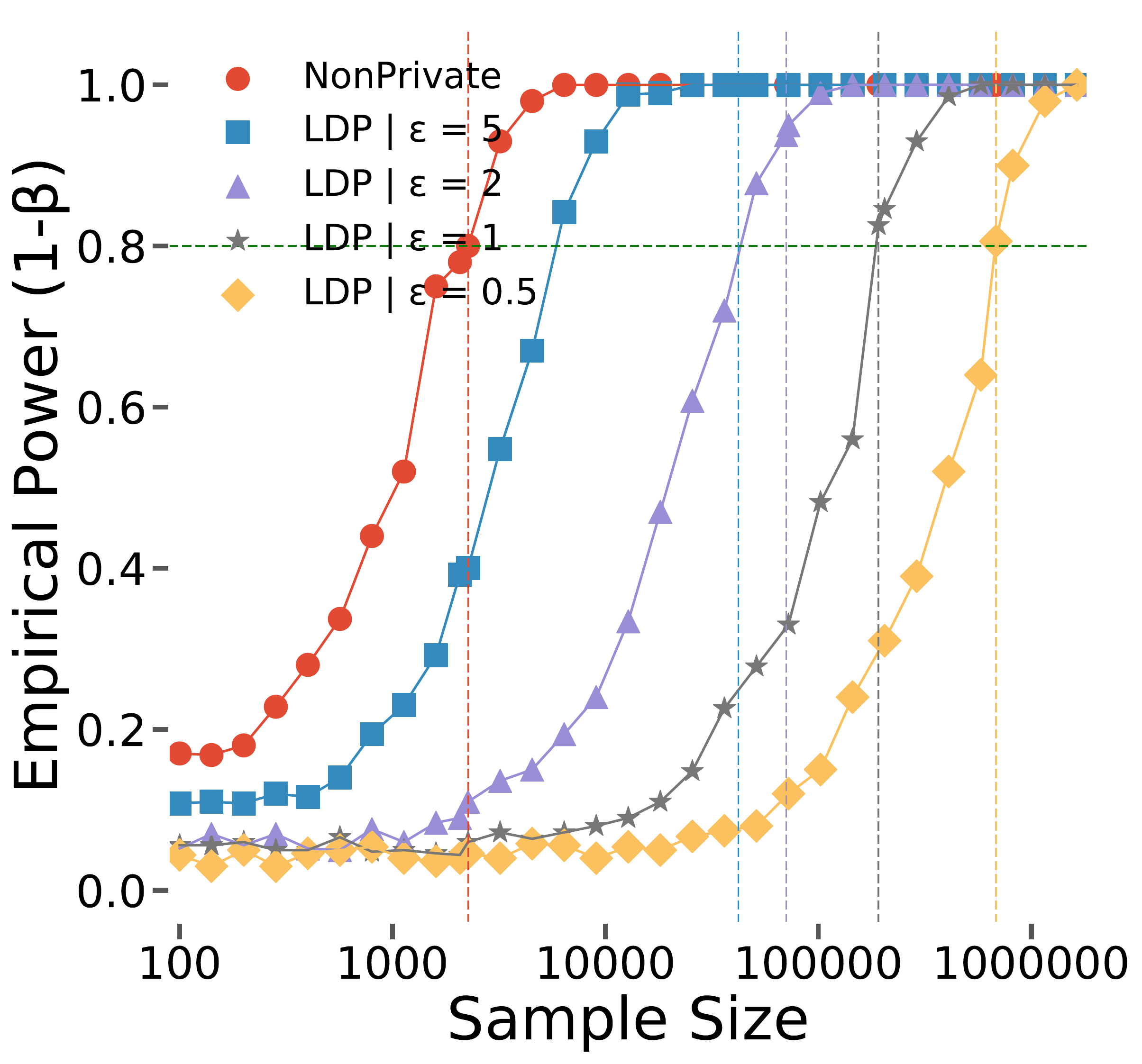}
}
\subfigure[Gap $\theta = 300$]{
\label{fig:exp:type2:300}
\includegraphics[height=3.9cm]{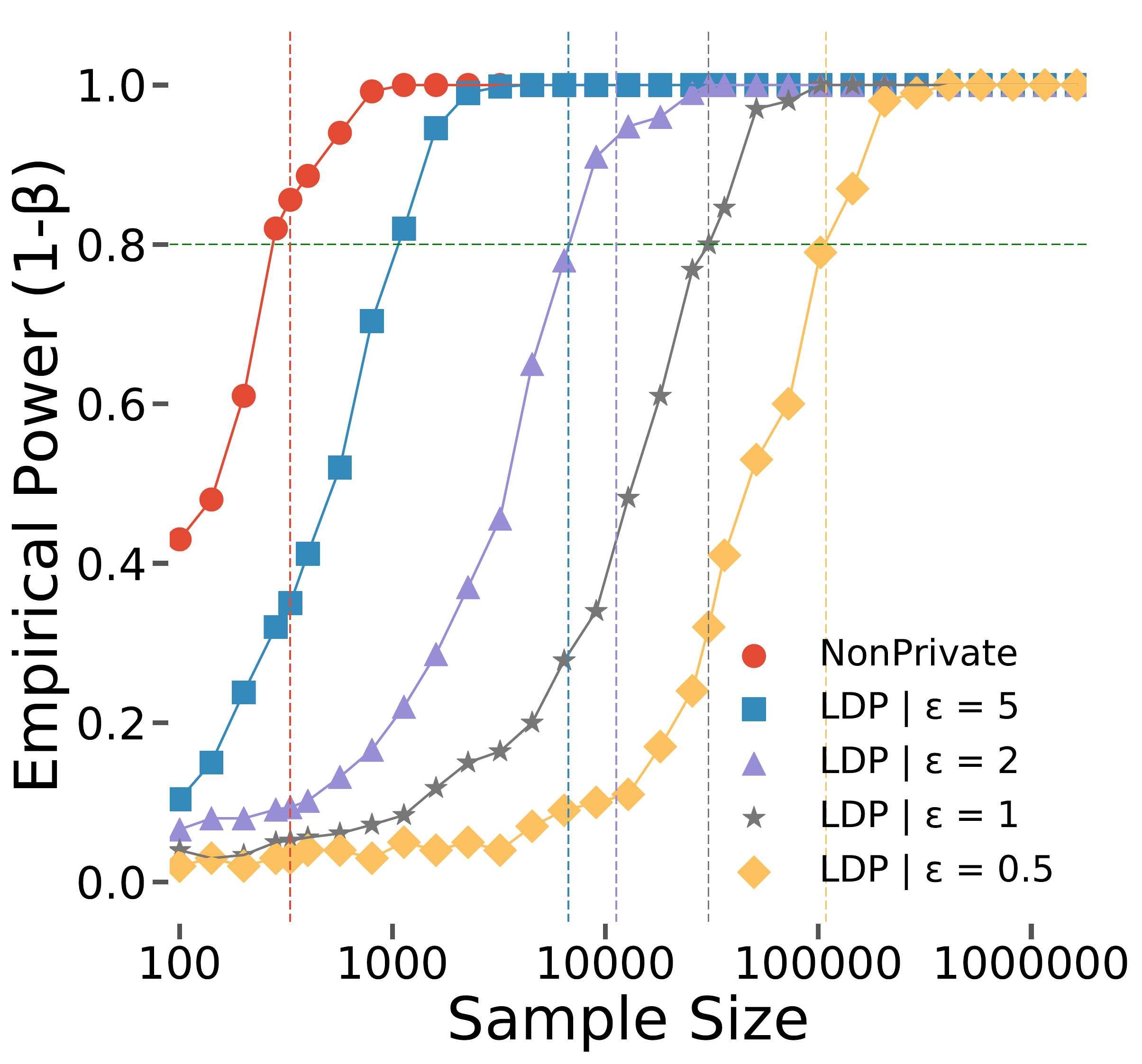}
}
\subfigure[Gap $\theta = 600$]{
\label{fig:exp:type2:600}
\includegraphics[height=3.9cm]{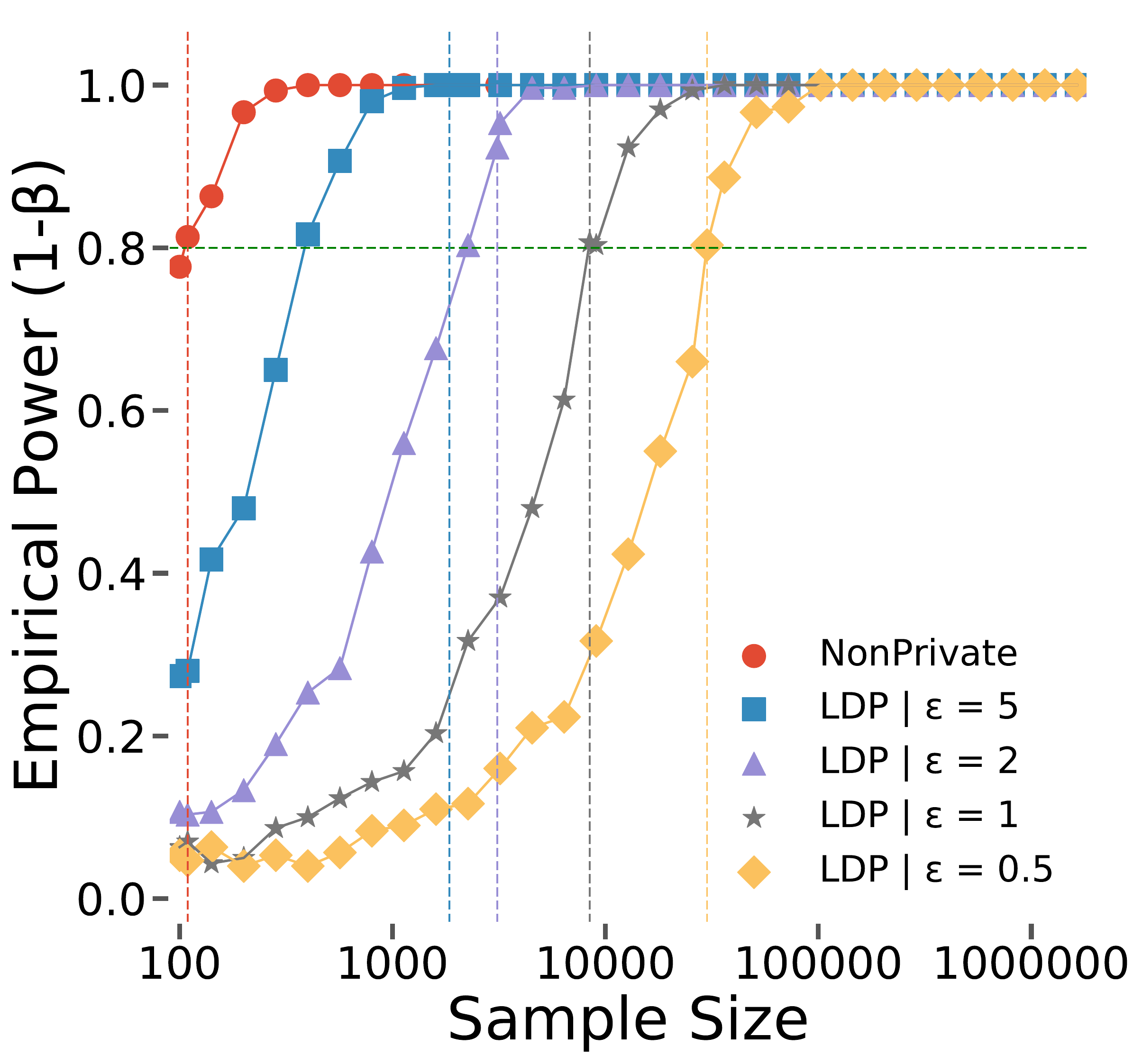}
}
%\vspace{-0.4cm}
\caption{Empirical power ($1$ $-$ type-II error) of different approaches in different scenarios for varying $\theta$ and $\eps$}
%\vspace{-0.4cm}
\label{fig:exp:type2}
\end{figure*}

\begin{figure}[t]
\centering
\includegraphics[width=8.3cm]{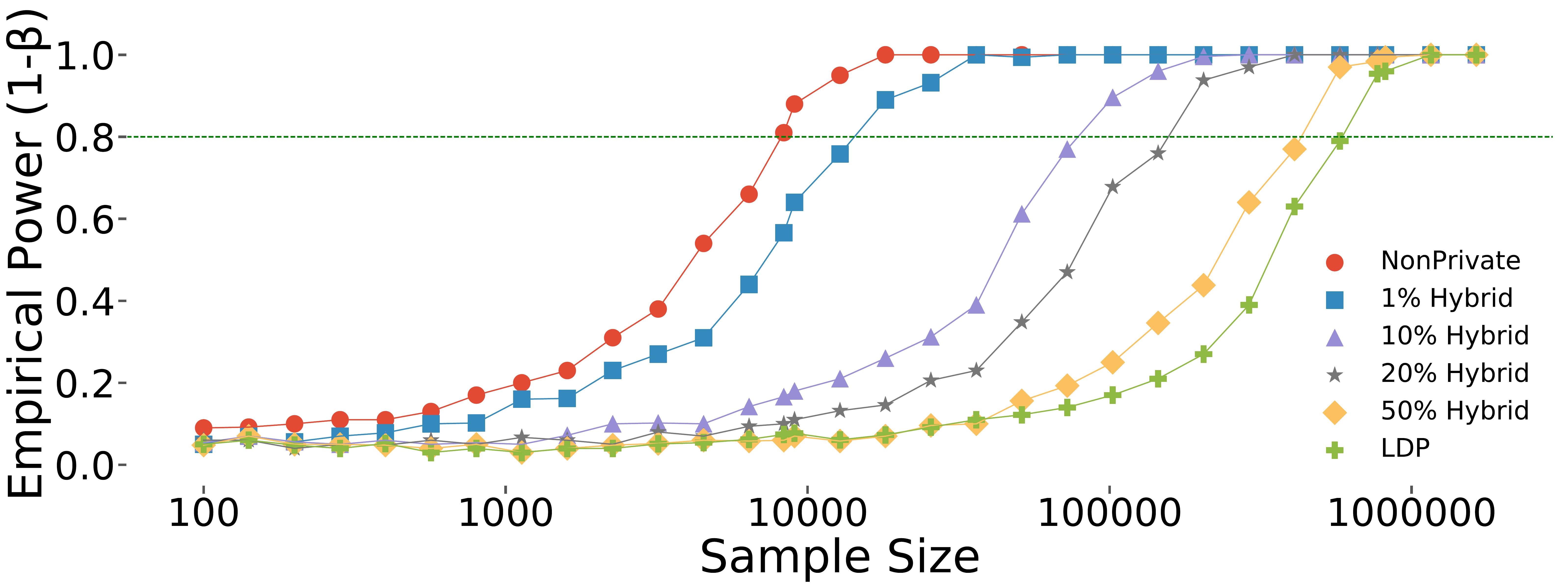}
%\vspace{-0.4cm}
\caption{Empirical power ($1$ $-$ type-II error) of $\testmix_{\eps}$ for hybrid privacy requirements with varying LDP fraction}
%\vspace{-0.4cm}
\label{fig:exp:type2:mix}
\end{figure}

\section{Experimental Evaluation}
\label{sec:exp}

% We evaluate and compare the LDP tests we proposed using a real-world dataset collected by a popular app.

\stitle{Dataset and parameters.} There are $20$ million users in this real-world dataset. Each user has a counter with the value in $[0,15000]$, \ie, $m = 15000$.
There are categorical attributes, \eg, country, associated with each user's counter.
% , and we use them to partition the population into sub-populations.
%
% We generate synthetic datasets from it by shifting the means of sub-populations.

We vary the privacy parameter $\eps$ from $0.5$ to $5$. The pre-specified significance level $\alpha = 0.05$, and the null hypothesis is $H_0: \mu_\pa - \mu_\pb = 0$. We draw samples from control/treatment with equal sizes $n_A = n_B$.
%
% In the context of hypothesis testing or A/B testing, the smaller the sample size needed the better.
%
For each value of sample sizes, we repeat drawing samples and conducting (LDP) tests $1000$ times, and report the average (empirical) type-I error and statistical power ($1$ $-$ type-II error).

\subsection{Evaluating Significance}
In the first set of experiments, we draw samples from distributions with the same means, and thus $H_0$ holds. We report the average type-I error of different approaches in \cfig\ref{fig:exp:type1}, \ie, the empirical probability that $H_0$ is rejected.

We conduct Welch's $t$-test on the non-privatized samples ({\sf NonPrivate} in \cfig\ref{fig:exp:type1}), estimation-based LDP test $\testest_{\eps_1,\eps_2}$ ({\sf 2bitLDP}), LDP test $\testbin_\eps$ in \calg\ref{alg:bittest} ({\sf LDP}), and $\testmix_\eps$ in \calg\ref{alg:bittest:mix} when $50\%$ users require LDP ({\sf 50\% Hybrid}).

Since it is required that type-I error $\leq$ significance level $\alpha$, ideally, the empirical type-I error should be close to $\alpha$ or less. \cfig\ref{fig:exp:type1} verifies \cthms\ref{thm:bittest:power}-\ref{thm:bittest:mix:power}: for different sample sizes and $\eps$, both $\testbin_\eps$ and $\testmix_\eps$ always have empirical type-I errors close to $0.05$. However, $\testest_\eps$ does not perform well, even with large sample sizes and large $\eps = \eps_1 + \eps_2 = 2.5+2.5$; the reason is that although the sample variance estimator \eqref{equ:varest} used in $\testest_\eps$ has bounded bias (\cprop\ref{prop:2bitvar}), its variance is too high as $m$ is large, and thus its empirical error has a non-negligible impact on the test statistic, leading to much higher type-I error than the pre-specified $\alpha$.

\subsection{Evaluating Power}
%
% Since even for large $\eps_1$ and $\eps_2$, $\testest_{\eps_1,\eps_2}$ gives statistically invalid results,
%
We will evaluate the power of $\testbin_\eps$ ({\sf LDP} in \cfig\ref{fig:exp:type2}) and $\testmix_\eps$ ({\sf x\% Hybrid} in \cfig\ref{fig:exp:type2:mix}) in the rest experiments, as they have been shown to satisfy the significance level empirically.

We draw samples from two populations with means differing by some constant $\theta$, and thus $H_0$ should be rejected. We report the average power in \cfigs\ref{fig:exp:type2}-\ref{fig:exp:type2:mix}, \ie, the empirical probability that $H_0$ is (correctly) rejected.

\stitle{Varying gap $\theta$ between control and treatment.}
We vary $\theta$ between two sub-populations from $60$ to $600$ in \cfigs\ref{fig:exp:type2:60}-\ref{fig:exp:type2:600}. In particular, $\theta = 60$ is a real scenario that we are comparing populations between two countries A and B.
We inject shifts to create synthetic cases with $\theta = 120$-$600$.

As $\theta$ grows, both $\testbin_\eps$ (for different $\eps$) and non-privatized Welch's $t$-test need smaller samples to achieve an empirical power $0.8$ (a threshold commonly used in practice). Intuitively, the larger the mean gap $\theta$ is, the easier it is for a test to distinguish the two populations. This trend is also consistent with theoretical bounds of the power derived in \cthm\ref{thm:bittest:power}.

For $\theta = 60$, to achieve an empirical power $0.8$, $\testbin_\eps$ (with $\eps = 5$) needs roughly three times as many samples as the non-privatized Welch's $t$-test does. It is totally acceptable, because, with the strong LDP privacy guaranteed for each user without the need of trusting the data collector, more users would be willing to share their data (in an LDP way).

\cthm\ref{thm:bittest:power} gives a way \eqref{equ:esnormal} to estimate the sample size needed to achieve certain power in $\testbin_\eps$. We can verify this analytical result here: in \cfigs\ref{fig:exp:type2:60}-\ref{fig:exp:type2:600}, for each $\eps$, we plot a dashed line in the same color as the color of the corresponding power curve -- this dashed line, calculated by \eqref{equ:esnormal}, denotes the sample size needed to achieve a power of $0.8$. It turns out to be a ``safe'' estimation: samples with this size always give the required or better power in \cfig\ref{fig:exp:type2}.

% \stitle{A real scenario.} We conduct tests to compare population means between two countries A and B and report results in \cfig\ref{fig:exp:type2:real}. The population means turn out to differ by $61.4$, so the result looks similar to \cfig\ref{fig:exp:type2:60}.

\stitle{Hybrid privacy requirements.} We evaluate our test $\testmix_\eps$ for hybrid privacy requirements in the scenario with $\theta = 60$ and $\eps = 1$ in \cfig\ref{fig:exp:type2:mix}: ``{\sf $x\%$ Hybrid}'' represents $\testmix_\eps$ on a population with a random portion of $x\%$ users requiring $1$-LDP, and ``{\sf LDP}'' represents $\testbin_\eps$. Intuitively, the less users require LDP, the easier it is for us to conduct tests. This intuition is consistent with the empirical performance of $\testmix_\eps$, which calibrates noise for each LDP user and mixes them with exact samples from those who do not require LDP.

Even within one sample, LDP users may follow a different distribution from non-LDP ones. So $\testmix_\eps$ still needs larger samples than the non-private test does. But the sample size needed to achieve a high power (\eg, $0.8$) decreases quickly as the LDP ratio goes down: when $50\%$ users requires LDP, the sample size needed in $\testmix_\eps$ is around half of the one in $\testbin_\eps$; and when $1\%$ users requires LDP, the sample size needed is close to the one in the non-private $t$-test.

% To empirically evaluate our experiment methdology and statistical computations, we performed two emprial evlatuions. In the first validation, we examined real-world usage data from xx users in between date date to compare difference when injecting differences (lifts). For differences of three difference sizes, we computed the sample sizes needed under standard methodologies and adjusted methodology. We then samples, a control and a treatment group, multiple times and conducted the t-test and computed the false positive and false negative rates. 

% In pink line the chart below, we see that with no LDP, the t-test with standard power computations perform exactly as expected, detecting differences 80\% of runs. However, when applying LDP and taking the translated mean, the t-test cannot accurate detect the difference, only xx \% of runs. With the adjusted power computation, the error is remedied. As seen in the subsequent lines, with the adjusted n, the t-test returns to decting differences at 80\% of runs. 

% To adjust for real-world conditions, we make the adjustment of taking the trimmed mean 99.5, since real-world contains outliers for anomalous conditions (\eg client clock problems, data corruption, etc.). This reduces the max m and gives a tigther bound for n. Assumption of  knowledge about the distribution of the data being collected is reasonable, as it is needed in normal power computations under normal conditions as well.

\section{Conclusion}
\label{sec:conclusion}

We study how to conduct hypothesis testing for comparing population means under LDP.
% Aiming at applications like A/B testing, we use LDP to provide a strong privacy guarantee for each user without needing them to trust the data collector (\eg, a tech company).
We propose two approaches. Both inject noise into each user's data in the samples before sending it to the data collector to ensure LDP. The first one, called estimation-based LDP test, decodes LDP samples aggregatively at the data collector to recover the observed test statistics. The second one, called transformation-based LDP test, studies the relationship between the population distributions and the distributions of LDP samples. It conducts transformed tests directly on LDP samples and converts conclusions for the original tests. The second one has provable significance and lower bounds of power, and it can be extended for population with hybrid privacy requirements.

{\small
\bibliographystyle{aaai}
\bibliography{Ding-ref}
}

%!TEX root = Ding.tex

\appendix

\section{Description of Estimation-based Test}

The estimation-based LDP test $\testest_{\eps_1, \eps_2}$ is described in \calg\ref{alg:esttest}. For the two samples $A$ and $B$, we use privacy budget $\eps_1$ to collects $A'$ and $B'$, and $\eps_2$ to collects $A''$ and $B''$ (lines~1-7). $\hat\mu_{\eps_1,m}(A')$ and $\hat\mu_{\eps_1,m}(B')$ are used to estimate sample means $\mu_A$ and $\mu_B$, respectively; $\hat s^2_{\eps_1, \eps_2, m}(A',A'')$ and $\hat s^2_{\eps_1, \eps_2, m}(B',B'')$ are used to estimate sample variances (lines~8-9).
The four estimations are embedded into \eqref{equ:tstatistic}, to estimate $t$ and $df$ in order to conduct a $t$-test with null $H_0$ (lines~10-12).
There is no theoretical guarantee on the errors in the testing result, mainly because it is difficult to bound the error in $\hat t$, and the distribution of $\hat t$ is unknown.
% But we have the following privacy guarantee.
%
\begin{algorithm}
{\bf Input:} Two samples $A = \{a_i\}_{i \in [n_A]}$ and $B = \{b_i\}_{i \in [n_B]}$.
\\
{\bf Null hypothesis $H_0$:} $\mu_\pa - \mu_\pb = d_0$.
\\
{\bf Parameters:} privacy budget $\eps_1,\eps_2$ and significance level $\alpha$. \\ \vspace{-0.4cm}
\begin{algorithmic}[1]
\STATE For users $i = 1$ to $n_A$ do
\STATE ~~~ Encode $a'_i = \algobit{\eps_1,m}(a_i)$ and $a''_i = \algobit{\eps_2,m^2}(a^2_i)$.
\STATE ~~~ Send $a'_i$ and $a''_i$ to the server.
\STATE For users $i = 1$ to $n_B$ do
\STATE ~~~ Encode $b'_i = \algobit{\eps_1,m}(b_i)$ and $b''_i = \algobit{\eps_2,m^2}(b^2_i)$.
\STATE ~~~ Send $b'_i$ and $b''_i$ to the server.
\STATE Server receives $A' = \{a'_i\}$, $A'' = \{a''_i\}$, $B'$, and $B''$.
\STATE Estimate $\hat\mu_A \! = \! \hat\mu_{\eps_1,m}(A')$ and $\hat s^2_A \! = \! \hat s^2_{\eps_1, \eps_2, m}(A', A'')$.
\STATE Estimate $\hat\mu_B \! = \! \hat\mu_{\eps_1,m}(B')$ and $\hat s^2_B \! = \! \hat s^2_{\eps_1, \eps_2, m}(B', B'')$.
\STATE Compute $\hat t = t(\hat\mu_A, \hat\mu_B, \hat s^2_A, \hat s^2_B, n_A, n_B)$ as in \eqref{equ:tstatistic}.
\STATE Compute $\hat{df} = df(\hat s^2_A, \hat s^2_B, n_A, n_B)$ as in \eqref{equ:tstatistic}.
\STATE Conduct a $t$-test with null hypothesis $H_0$, using $\hat t$ and $\hat{df}$ as $t$ and $df$ at significance level $\alpha$.
\end{algorithmic}
\caption{$\testest_{\eps_1, \eps_2}$: Estimation-based LDP Test}
\label{alg:esttest}
\end{algorithm}

\section{Proof to \cprop\ref{prop:1bitdp}}

This original proof is from \cite{nips:DingKY17}, and we include it here for the completeness.

Each user sends only one bit $\algobit{\eps,m}(x)$ to the data collector if her/his counter is $x$. According to \eqref{equ:1bitmech}, for any $x \in [0,m]$, the probability of sending $0$ varies from ${1}/({e^\eps+1})$ to ${e^\eps}/({e^\eps+1})$; similarly, the probability of sending $1$ is in the same range. So the ratio between respective probabilities for any two different values of $x$ is at most $e^\eps$. So the privacy guarantee $\algobit{\eps,m}$ follows from \cdef\ref{def:ldp}.
\qed

\section{Proof to \cprop\ref{prop:varhardness}}

Suppose $X = \{x_i\}$ is a random sample drawn from a distribution $\px$ such that $\pr{x_i = m/2} = 1$, and $Y = \{y_i\}$ is a sample drawn from $\py$ such that $\pr{y_i = 0} = 1/2$ and $\pr{y_i = m} = 1/2$. So $X$ has a sample variance $s^2_X = 0$, while with high probability, $Y$ has a sample variance at least $s^2_Y = \bigomega{m^2}$. However, the LDP bits $\algobit{\eps,m}(X) = \{\algobit{\eps,m}(x_i)\}$ and $\algobit{\eps,m}(Y)$ follows the same distribution:
\begin{align*}
\forall x_i \in X: &~ \pr{\algobit{\eps,m}(x_i) = 0} = \pr{\algobit{\eps,m}(x_i) = 1} = 1/2;
\\
\forall y_i \in Y: &~ \pr{\algobit{\eps,m}(y_i) = 0} = \pr{\algobit{\eps,m}(y_i) = 1} = 1/2.
\end{align*}
The conclusion follows from the fact that $X$ and $Y$ cannot be distinguished from each other based on their $\eps$-LDP collections $\algobit{\eps,m}(X)$ and $\algobit{\eps,m}(Y)$, while their sample variances have a $\bigomega{m^2}$ gap with high probability.
\qed

\section{Proof to \cprop\ref{prop:2bitvar}}

Using \cprop\ref{prop:1bitmean} on $X''$, we have
\[
\ep{\hat\mu_{\eps_2,m^2}(X'')} = \mu_{X^2};
\]
and on $X'$, we have
\[
\vr{\hat\mu_{\eps_1,m}(X')} = \bigoh{\frac{m^2}{n\eps_1}}.
\]
So using the linearity of expectation and \eqref{equ:var}-\eqref{equ:varest}, i) follows:
\begin{align}
& \ep{\hat s^2_{\eps_1, \eps_2, m}(X', X'')} \nonumber
\\
= & \frac{n}{n-1}\left(\ep{\hat\mu_{\eps_2,m^2}(X'')} - \ep{\hat\mu^2_{\eps_1,m}(X')}\right) \nonumber
\\
= & \frac{n}{n-1}\left(\mu_{X^2} - \left(\vr{\hat\mu_{\eps_1,m}(X')} + \ep{\hat\mu_{\eps_1,m}(X')}^2\right) \right) \nonumber
\\
= & \frac{n}{n-1}\left(\mu_{X^2} - \mu_X^2\right) - \frac{n}{n-1}\vr{\hat\mu_{\eps_1,m}(X')} \nonumber
\\
= & s_X^2 - \bigoh{\frac{m^2}{n\eps_1}}. \nonumber
\end{align}

To bound the variance of $\hat s^2_{\eps_1, \eps_2, m}(X', X'')$ as in ii), because $X'$ and $X''$ are generated independently, we have
\begin{align}
& \vr{\hat s^2_{\eps_1, \eps_2, m}(X', X'')} \label{equ:varestvar1}
\\
= & \frac{n^2}{(n-1)^2}\left(\vr{\hat\mu_{\eps_2,m^2}(X'')} + \vr{\hat\mu^2_{\eps_1,m}(X')}\right). \nonumber
\end{align}
From \cprop\ref{prop:1bitmean} and how each $x_i''$ is generated, we have
\begin{equation} \label{equ:varestvar2}
\vr{\hat\mu_{\eps_2,m^2}(X'')} = \bigoh{\frac{(m^2)^2}{n\eps_2^2}} = \bigoh{\frac{m^4}{n\eps_2^2}}.
\end{equation}
And from \eqref{equ:meanest}, we have
\begin{align}
& \vr{\hat\mu^2_{\eps_1,m}(X')} = \vr{\left(\frac{m}{n}\sum_{i=1}^{n} \frac{x'_i \cdot (e^{\eps_1}+1)-1}{e^{\eps_1}-1}\right)^2} \nonumber
\\
& = \frac{m^4}{n^4} \cdot \vr{\left(\sum_{i=1}^{n} \frac{x'_i \cdot (e^{\eps_1}+1)-1}{e^{\eps_1}-1}\right)^2} \nonumber
\\
& = \frac{m^4}{n^4} \cdot n^2 \cdot \bigoh{1} \cdot \left(\frac{e^{\eps_1}+1}{e^{\eps_1}-1}\right)^4 = \bigoh{\frac{m^4}{n^2\eps_1^4}}. \label{equ:varestvar3}
\end{align}
Putting \eqref{equ:varestvar2} and \eqref{equ:varestvar3} back to \eqref{equ:varestvar1}, ii) follows.
\qed

\section{The First Part of Proof to \cthm\ref{thm:bittest:power}}

We first need to derive a threshold value $z_0$ to determine the regions of acceptance and rejection for the test statistic $z(A^{\sf bin}, B^{\sf bin})$.
%
% Suppose $\mu_\pa = \mu_\pb$ and accordingly $p(\mu_\pa) = p(\mu_\pb)$,
%
Under $H_0: \mu_\pa - \mu_\pb = 0$, or equivalently, $H_0^{\sf bin}: p_\pa - p_\pb = 0$, we have $\ep{z(A^{\sf bin}, B^{\sf bin})} = 0$. And flipping the value of any of $a_i^{\sf bin} \in A^{\sf bin}$ (or $b_i^{\sf bin} \in B^{\sf bin}$) changes the value of $z(A^{\sf bin}, B^{\sf bin})$ by at most $1/n_A$ (or $1/n_B$). So from McDiarmid's inequality, we can bound
\begin{align}
\pr{z(A^{\sf bin}, B^{\sf bin}) \geq t} & \leq \exp\left(-\frac{2t^2}{n_A \cdot \frac{1}{n_A^2} + n_B \cdot \frac{1}{n_B^2}}\right) \nonumber
\\
& = \exp\left(-\frac{2t^2}{1/n_A + 1/n_B}\right). \label{equ:bittest:tail}
\end{align}
A significance level of $\alpha$ requires a threshold $z_0$, \st, under $H_0^{\sf bin}$, $\pr{\text{\rm reject}}$ $=$ $\pr{z(A^{\sf bin}, B^{\sf bin}) \geq z_0} \leq \alpha$.
Therefore, based on \eqref{equ:bittest:tail}, $H^{\sf bin}_0$ should be rejected iff:
\[
z(A^{\sf bin}, B^{\sf bin}) \geq z_0 = \sqrt{\left(\frac{1}{2n_A} + \frac{1}{2n_B}\right)\ln\frac{1}{\alpha}}.
\]

Now let's estimate the statistical power under the alternative hypothesis with $\mu_\pa - \mu_\pb = \theta$, or equivalently,
\begin{equation}
p_\pa - p_\pb = \ep{z(A^{\sf bin}, B^{\sf bin})} = \frac{\theta}{m} \cdot \frac{e^\eps-1}{e^\eps+1}. \label{equ:bittest:alternative}
\end{equation}
The statistical power is the probability of rejection
\[
P(\theta) = \pr{z(A^{\sf bin}, B^{\sf bin}) \geq \sqrt{\left(\frac{1}{2n_A} + \frac{1}{2n_B}\right)\ln\frac{1}{\alpha}}},
\]
under \eqref{equ:bittest:alternative}. Let $\Delta = \ep{z(A^{\sf bin}, B^{\sf bin})} - z(A^{\sf bin}, B^{\sf bin})$:
\begin{align}
& P(\theta) = \pr{\Delta \leq \frac{\theta}{m} \cdot \frac{e^\eps-1}{e^\eps+1} - \sqrt{\left(\frac{1}{2n_A} + \frac{1}{2n_B}\right)\ln\frac{1}{\alpha}}} \nonumber
\\
& \text{(using McDiarmid's inequality, again)} \nonumber
\\
\geq & 1 - \exp\left(-\frac{2\left(\frac{\theta}{m} \cdot \frac{e^\eps-1}{e^\eps+1} - \sqrt{\left(\frac{1}{2n_A} + \frac{1}{2n_B}\right)\ln\frac{1}{\alpha}}\right)^2}{\frac{1}{n_A} + \frac{1}{n_B}}\right) \nonumber
\\
\geq & 1 - \exp\left(-{\left(\frac{\theta}{m} \cdot \frac{e^\eps-1}{e^\eps+1} \cdot \sqrt{\frac{2n_An_B}{n_A+n_B}}- \sqrt{\ln\frac{1}{\alpha}}\right)^2}\right). \nonumber
\end{align}
Note that a pre-condition here is that 
\begin{align*}
& \frac{\theta}{m} \cdot \frac{e^\eps-1}{e^\eps+1} \cdot \sqrt{\frac{2n_An_B}{n_A+n_B}}- \sqrt{\ln\frac{1}{\alpha}} \geq 0
\\
\Longleftrightarrow ~~~ & n_A + n_B = \bigomega{\frac{m^2}{\theta^2 \eps^2}\ln\frac{1}{\alpha}}.
\end{align*}
The upper bound in \eqref{equ:sp0} holds under the above condition.
\qed

\end{document}